\pdfoutput=1
\documentclass[a4paper,UKenglish,thm-restate,cleveref, autoref]{lipics-v2021}

\usepackage{xcolor}
\usepackage{framed}
\colorlet{shadecolor}{blue!20}

\nolinenumbers

\newtheorem{fact}[theorem]{Fact}
\numberwithin{theorem}{section}
\numberwithin{lemma}{section}
\numberwithin{proposition}{section}
\numberwithin{corollary}{section}
\numberwithin{fact}{section}
\numberwithin{remark}{section}


\newcommand{\bool}{\{0,1\}}
\newcommand{\boolset}[1]{\bool^{#1}}
\newcommand{\boolfn}[1]{\bool^{#1}\longrightarrow \bool}
\newcommand{\Rank}{\textrm{Rank}}

\newcommand{\Value}{\textrm{Value}}
\newcommand{\AValue}{\textrm{ASym-Value}}
\newcommand{\SearchR}{\textrm{SearchR}}

\newcommand{\Depth}{\textrm{Depth}}
\newcommand{\DepthAO}{\textrm{Depth}_{\wedge,\vee}}
\newcommand{\DepthGA}{\textrm{Depth}_{\bar{\wedge}}}
\newcommand{\wDepth}{\textrm{Depth}_w}
\newcommand{\Size}{\textrm{DTSize}}
\newcommand{\dnf}{\textrm{DNF}}
\newcommand{\spar}{\textrm{spar}}
\newcommand{\sparn}{\tilde{\textrm{spar}}}%non-empty set coeffs
\newcommand{\skill}{\textrm{K}}
\newcommand{\codim}{\textrm{co-dim}}
\newcommand{\CertZ}{\textrm{C}_0}
\newcommand{\CertO}{\textrm{C}_1}
\newcommand{\minCert}{\textrm{C}_{\min}}

\newcommand{\aCert}{\textrm{C}_{avg}}
\newcommand{\Cert}{\textrm{C}}

\newcommand{\Gap}{\textrm{Gap}}
\newcommand{\minGap}{\textrm{Gap}_{\min}}
\newcommand{\OR}{\mbox{{\sc Or}}}
\newcommand{\AND}{\mbox{{\sc And}}}
\newcommand{\CONJ}{\mbox{{\sc Conj}}}
\newcommand{\Parity}{\mbox{{\sc Parity}}}
\newcommand{\Tribes}{\mbox{{\sc Tribes}}}
\newcommand{\dTribes}{\mbox{{\sc Tribes}}^d}
\newcommand{\Maj}{\mbox{{\sc Maj}}}
\newcommand{\Thr}{\mbox{{\sc Thr}}}

\title{On (Simple) Decision Tree Rank}
\author{Yogesh Dahiya}{The Institute of Mathematical Sciences, Chennai, (A CI of Homi Bhabha National Institute HBNI), India}{yogeshdahiya@imsc.res.in}{}{}
\author{Meena Mahajan}{The Institute of Mathematical Sciences (HBNI), Chennai, India}{meena@imsc.res.in}{https://orcid.org/0000-0002-9116-4398}{}

\authorrunning{Y. Dahiya and M. Mahajan}

\hideLIPIcs

\begin{document}

\maketitle

\begin{abstract}
In the decision tree computation model for Boolean functions, the
depth corresponds to query complexity, and size corresponds to storage
space. The depth measure is the most well-studied one, and is known to
be polynomially related to several non-computational complexity
measures of functions such as certificate complexity. The size measure
is also studied, but to a lesser extent. Another decision tree measure
that has received very little attention is the minimal rank of the
decision tree, first introduced by Ehrenfeucht and  Haussler in 1989.
This measure is closely related to the logarithm of the size, but is not polynomially related to depth, and hence it can reveal additional
information about the complexity of a function. 
It is characterised by the value of a Prover-Delayer game first proposed
by Pudl{\'a}k and Impagliazzo
in the context of tree-like resolution proofs.

In this paper we
study this measure further. We obtain an upper bound on depth in terms of rank and Fourier sparsity. We obtain upper and lower bounds on rank in
terms of (variants of) certificate complexity. We also obtain upper
and lower bounds on the rank for composed functions in terms of the
depth of the outer function and the rank of the inner function. This allow us to easily recover known asympotical  lower bounds on logarithm of the size for Iterated AND-OR and Iterated 3-bit Majority. We compute the rank exactly for several natural functions and use them to
show that all the bounds we have obtained are tight. We also show that rank in the  simple decision tree model can be used to bound query complexity, or depth, in the more general conjunctive decision tree model.  
Finally, we improve upon the known size lower bound for the Tribes function and conclude  that in
the size-rank relationship for decision trees, obtained
by Ehrenfeucht and  Haussler, the upper bound for Tribes is asymptotically tight. 
\end{abstract}

\section{Introduction}
\label{sec:intro}
The central problem in Boolean function complexity is to understand
exactly how hard it is to compute explicit functions. The hardness
naturally depends on the computation model to be used, and depending
on the model, several complexity measures for functions have been
studied extensively in the literature. To name a few --  size and depth
for circuits and formulas, size and width for
branching programs, query complexity, communication complexity, length
for span programs, and so on. All of these are measures of the
computational hardness of a function.  There are also several ways to
understand hardness of a function intrinsically, independent of a
computational model. For instance, the sensitivity of a function, its
certificate complexity, the sparsity of its Fourier spectrum, its
degree and approximate degree, stability, and so on. Many bounds on
computational measures are obtained by directly relating them to
appropriate intrinsic complexity measures. See \cite{Jukna-BFCbook}
for a wonderful overview of this area. Formal definitions of relevant
measures appear in \cref{sec:prelim}.

Every Boolean function $f$ can be computed by a simple decision tree (simple in the sense that each node queries a single variable),
which is one of the simplest computation models for Boolean
functions. The most interesting and well-studied complexity measure in
the decision tree model is the minimal depth $\Depth(f)$, measuring
the query complexity of the function. This measure is known to be
polynomially related to several intrinsic measures: sensitivity, block
sensitivity, certificate complexity.  But there are also other
measures which reveal information about the function. The minimal size
of a decision tree, $\Size(f)$, is one such measure, which measures the storage
space required to store the function as a tree, and has received some
attention in the past.

A measure which has received relatively less attention is the minimal rank
of a decision tree computing the function, first defined and studied
in \cite{EH-IC1989}; see also \cite{ABDORU10}. In general, the rank of
a rooted tree (also known as its Strahler number, or Horton-Strahler
number, or tree dimension) measures its branching complexity, and is a
tree measure that arises naturally in a wide array of applications;
see for instance \cite{EsparzaLS-LATA14}. The rank of a
Boolean function $f$, denoted $\Rank(f)$, is the minimal rank of a decision tree computing
it. The original motivation for considering rank of decision trees was
from learning theory -- an algorithm, proposed in \cite{EH-IC1989},
and later simplified in \cite{Blum92}, shows that constant-rank decision
trees are efficiently learnable in Valiant's PAC learning framework \cite{Valiant-CACM84}. Subsequently, the rank measure has played an important role in
understanding the decision tree complexity of search problems over
relations \cite{pudlak2000lower,esteban2003combinatorial,Kullmann-ECCC-99} -- see more in the Related Work part below.  The special
case when the relation corresponds to a Boolean function is exactly
the rank of the function.  However, there is very little work focussing
on the context of, and exploiting the additional information from,
this special case. This is precisely the topic of this paper.

In this paper, we study how the rank of boolean functions relates
to other measures.  In contrast with $\Depth(f)$, $\Rank(f)$ is not
polynomially related with sensitivity or to certificate complexity
$\Cert(f)$, although it is bounded above by $\Depth(f)$.
Hence it can reveal additional information about the complexity of
a function over and above that provided by $\Depth$.  For instance, from several viewpoints, the $\Parity_n$ function is significantly harder than the $\AND_n$ function. But both of them have the same $\Depth$, $n$. However, $\Rank$ does reflect this difference in hardness, with $\Rank(\AND_n)=1$ and $\Rank(\Parity_n)=n$. 
On the other hand, rank is also already known
to characterise the logarithm of decision tree size (\Size), upto a $\log n$
multiplicative factor. 
Thus lower bounds on rank  give lower bounds on the space required to store a decision tree explicitly. (However, the $\log n$ factor is crucial; there is no dimension-free characterisation. Consider e.g.\ $\log \Size(\AND_n) = \Theta(\log n)$.) 

Our main findings can be summarised as follows:
\begin{enumerate}
  \item $\Rank(f)$ is equal to the value of the Prover-Delayer game of
    Pudl{\'a}k and Impagliazzo \cite{pudlak2000lower} played on the
    corresponding relation $R_f$ (\cref{thm:game-rank}). (This is implicit in earlier literature \cite{Kullmann-ECCC-99,esteban2003combinatorial}.)
  \item While $\Rank$ alone cannot give upper bounds on $\Depth(f)$,
    $\Depth(f)$ is bounded above by the product of $\Rank(f)$ and
    $1+\log\spar(f)$ (\cref{thm:rank-sparsity-depth}).
  \item $\Rank(f)$ is bounded between the minimum certificate
    complexity of $f$ at any point, and $(\Cert(f)-1)^2+1$;
    \cref{thm:rank-cert-bounds}. The upper bound
    (\cref{lem:rank-cert}) is an improvement on the bound inherited
    from $\Depth(f)$, and is obtained by adapting that construction.
  \item For a composed function $f\circ g$, $\Rank(f\circ g)$ is 
    bounded above and below by functions of $\Depth(f)$ and
    $\Rank(g)$; \cref{thm:compose-rank-bounds}. The main technique in
    both bounds (\cref{thm:compose-rank-ub,thm:compose-rank-lb}) is to
    use weighted  decision trees, as was used in the
    context of depth \cite{Montanaro-cj14}. For iterated composed functions,
    these bounds can be used recursively (\cref{corr:iterated-rank}), and
    can be used to easily recover known bounds on $\Rank$ for some functions
    (\cref{corr:examples}).
  \item The measures $\Rank$ and $\log\Size$ for simple decision trees
    sandwich the query complexity in the more general decision tree
    model where each node queries a conjunction of literals
    (\cref{thm:simple-conj-relation}).
  \item In the relation between $\Rank(f)$ and $\Size(f)$ from
    \cite{EH-IC1989}, the upper bound on $\log\Size$ is asymptotically
    tight for the $\Tribes$ function (\cref{sec:rank-size}).
\end{enumerate}
By calculating the exact rank for specific functions, we show that all
the bounds we obtain on rank are tight. We also describe optimal strategies for the Prover and Delayer, for those more familiar with that setting.

\paragraph*{Related work.}
A preliminary version of this paper, with some proofs omitted or
only briefly sketched, appears in the proceedings of the FSTTCS 2021
conference \cite{DM21}.

In \cite{KS-JCSS04} (Corollary 12), non-trivial learnability of
$s$-term DNFs is demonstrated.  The crucial result that allows this
learning is the transformation of the DNF expression into a polynomial
threshold function of not too large degree.  An important tool in the
transformation is the rank of a hybrid kind of decision tree; in these
trees, each node queries a single variable, while the subfunctions at
the leaves, though not necessarily constant, have somewhat small
degree.  The original DNF is fist converted to such a hyrid tree with
a bound on its rank, and this is exploited to achieve the full
conversion to low-degree polynomial threshold functions. This
generalises an approach credited in \cite{KS-JCSS04} to Satya Lokam. 

In \cite{ABDORU10}, a model called $k^+$-decision trees is considered, and the complexity is related to both simple decision tree rank and to communication complexity. In particular, Theorems 7 and 8 from \cite{ABDORU10} imply that communication complexity lower bounds with respect to any variable partition (see \cite{KN-CC-book}) translate to decision tree rank lower bounds, and hence by \cite{EH-IC1989} to decision tree size lower bounds. 

In \cite{TV97}, the model of linear decision trees is considered (here each node queries not a single variable but a linear threshold function of the variables), and for such trees of bounded rank computing the inner product function, a lower bound on depth is obtained. Thus for this function, in this model, there is a trade-off between rank and depth. In \cite{UT2015}, rank of linear decision trees is used in obtaining non-trivial upper bounds on depth-2 threshold circuit size. 

In \cite{pudlak2000lower}, a 2-player game is described, on an
unsatisfiable formula $F$ in conjunctive normal form, that constructs
a partial assignment falsifying some clause. The players are referred
to in subsequent literature as the Prover and the Delayer.  The value
of the game, $\Value(F)$, is the maximum $r$ such that the Delayer can
score at least $r$ points no matter how the Prover plays. It was shown
in \cite{pudlak2000lower} that the size of any tree-like resolution
refutation of $F$ is at least $2^{\Value(F)}$.  Subsequently, the
results of \cite{Kullmann-ECCC-99,esteban2003combinatorial} yield the
equivalence $\Value(F) = \Rank(F)$, where $\Rank(F)$ is defined to be
the minimal rank of the tree underlying a tree-like resolution
refutation of $F$. (Establishing this equivalence uses
refutation-space and tree pebbling as intermediaries.)  The relevance
here is because there is an immediate, and well-known, connection to
decision trees for search problems over relations: tree-like
resolution refutations are decision trees for the corresponding search
CNF problem. (See Lemma 7 in \cite{BIW04}).  Note that the size lower
bound from \cite{pudlak2000lower}, and the rank-value equivalence from
\cite{Kullmann-ECCC-99,esteban2003combinatorial}, hold for the 
search problem over arbitrary relations, not just searchCNF. (See e.g.\ 
Exercise 14.16 in Jukna for the size bound.) In particular, for
Boolean function $f$, it holds for the corresponding canonical
relation $R_f$ defined in \cref{sec:prelim}. Similarly, the value of an asymmetric variant of this game is known to characterise the size of a decision tree for the search CNF problem \cite{BGL13}, and this too holds for general relations and Boolean functions. 

%For completeness, we give, in \cref{sec:game}, a self-contained proof of the rank-value equivalence for arbitrary relations. In the rest of the paper, we present our bounds on $\Rank(f)$ by direct inductive arguments/decision tree constructions. They can also be stated using the equivalence of the game value and rank -- while this does not particularly simplify the proofs, it changes the language of the proofs and may be more accessible to the reader already familiar with that setting.  We include game-based proofs as a separate section at the end. 

\paragraph*{Organisation of the paper.}
After presenting basic definitions and known results in
\cref{sec:prelim}, we describe the Prover-Delayer game from
\cite{pudlak2000lower} in \cref{sec:game}, and observe that its value
equals the rank of the function. We also describe the asymmetric game from \cite{BGL13}. 
We compute the rank of some simple
functions in \cref{sec:simple-calc}. In \cref{sec:rank-rels}, we
describe the relation between rank, depth, Fourier sparsity,  and certificate complexity. In
\cref{sec:rank-composed}, we present results concerning composed
functions. In \cref{sec:application} we give two applications. Firstly, using our  rank lower bound result, we prove the tight $\log$ size lower bound. Secondly, we   prove a query lower bound in the $\CONJ$ decision tree model. In \cref{sec:rank-size} we examine the size-rank relationship for the $\Tribes$ function.
The bounds in \cref{sec:simple-calc,sec:rank-rels,sec:rank-composed,sec:rank-size} are all obtained by direct inductive
arguments/decision tree constructions. They can also be stated using
the equivalence of the game value and rank -- while this does not
particularly simplify the proofs, it changes the language of the
proofs and may be more accessible to the reader already familiar with
that setting. Hence we include such game-based arguments for our results in \cref{sec:game-proofs}.

\section{Preliminaries}
\label{sec:prelim}
\paragraph*{Decision trees}
For a Boolean function $f: \boolfn{n}$, a decision tree computing $f$ is a binary tree with internal nodes labeled by the variables and the leaves labelled by $\bool$. To evaluate a function on an unknown input, the process starts at the root of the decision tree and works down the tree, querying the variables at the internal nodes. If the value of the query is $0$, the process continues  in the the left subtree, otherwise it proceeds in the right subtree. The label of the leaf so reached is the value of the function on that particular input. A decision tree is said to be reduced if no variable is queried more than once on any root-to-leaf path. Without loss of generality, any decision tree can be reduced, so in our discussion, we will only consider reduced decision trees. The depth $\Depth(T)$ of a decision tree $T$ is the length of the longest root-to-leaf path,  and its size $\Size(T)$ is the number of leaves. The decision tree complexity or the depth of $f$, denoted by $\Depth(f)$, is defined to be the minimum depth of a decision tree computing $f$. Equivalently, $\Depth(f)$ can also be seen as the minimum number of worst-case queries required to evaluate $f$. The size of a function $f$, denoted by $\Size(f)$, is defined similarly i.e.\ the minimum size of a decision tree computing $f$. Since decision trees can be reduced, $\Depth(f) \le n$ and $\Size(f) \le 2^n$ for every $n$-variate function $f$. A function is said to be evasive  if its depth is maximal, $\Depth(f)=n$. 

\paragraph*{Weighted decision trees}
Weighted decision trees
describe query complexity in settings where querying different
input bits can have differing cost, and arises naturally in the
recursive construction. 
Formally, these are defined as follows: Let $w_i$ be the cost of querying variable $x_i$. For a decision tree $T$, its weighted depth with respect to the weight vector $[w_1,\ldots,w_n]$, denoted by $\wDepth(T,[w_1,w_2,...,w_n] )$,  is the maximal sum of weights of the variables specified by the labels of nodes of $T$ on any root-to-leaf  path. The weighted decision tree complexity of $f$, denoted by $\wDepth(f,[w_1,w_2,...,w_n] )$, is the minimum weighted depth of a decision tree computing $f$. Note that $\Depth(f)$ is exactly $\wDepth(f,[1,1,\ldots ,1])$. The following fact is immediate from the definitions.
\begin{fact}\label{fact:wtd-dec-tree}
  For any reduced decision tree $T$ computing an $n$-variate function,  weights $w_1, \ldots, w_n$, and  $i\in [n]$,
  \[
  \wDepth(T,[w_1,\ldots , w_{i-1}, w_i+1, w_{i+1}, \ldots   ,w_n] ) \le
  \wDepth(T,[w_1,w_2,...,w_n] ) +1.
  \]
\end{fact}

\paragraph*{Certificate Complexity}
The certificate complexity of a function $f$, denoted $\Cert(f)$, measures the number of variables that need to be assigned in the worst case to fix the value of $f$. More precisely, for a Boolean function $f:\boolfn{n}$ and an input $a\in \boolset{n}$, an $f$-certificate of $a$ is a subset $S \subseteq \{1,...,n\}$ such that the value of $f(a)$ can be determined by just looking at the bits of $a$ in set $S$. Such a certificate need not be unique. Let $\Cert(f,a)$ denote the minimum size of an $f$-certificate for the input $a$. That is,
\[\Cert(f,a) = \min\left\{ |S| \mid S\subseteq [n]; \forall a'\in \bool^n,
\left[\left(a'_j=a_j \forall j\in S\right) \implies f(a')=f(a)\right]\right\}. \]
Using this definition, we can define several measures.
\begin{align*}
 \textrm{For ~}b\in\bool,~~ \Cert_b(f) & = \max\{ \Cert(f,a) \mid  a\in f^{-1}(b)\} \\
  \Cert(f) & = \max\{ \Cert(f,a) \mid  a\in \bool^n\} = \max\{\Cert_0(f),\Cert_1(f) \}\\
  \aCert(f) & = 2^{-n}\sum_{a\in \bool^n} \Cert(f,a) \\
  \minCert(f) & = \min\{ \Cert(f,a) \mid  a\in \bool^n\} 
\end{align*}

\paragraph*{Composed functions}
For boolean functions $f, g_1,g_2,\ldots ,g_n$ of arity $n, m_1, m_2, \ldots , m_n$ respectively, the composed function $f\circ(g_1,g_2,...,g_n)$ is a function of arity $\sum_i m_i$, and is defined as follows: for $a^i\in \bool^{m_i}$ for each $i\in n$, $f\circ(g_1,g_2,...,g_n)(a^1,a^2,...,a^n)=f(g_1(a^1),g_2(a^2),\ldots ,g_n(a^n))$. We call $f$ the outer function and $g_1,\ldots ,g_n$ the inner functions. For functions $f:\boolfn{n}$ and $g:\boolfn{m}$, the composed function $f\circ g$ is the function $f\circ (g,g,\ldots ,g):\boolfn{mn}$. The composed function $\OR_n\circ\AND_m$ has a special name, $\Tribes_{n,m}$, and when $n=m$, we simply write $\Tribes_n$.  Its dual is the function $\AND_n\circ\OR_m$ that we denote $\dTribes_{n,m}$. (The dual of $f(x_1, \ldots, x_n)$ is the function $\neg f(\neg x_1, \ldots, \neg x_n)$.)

%\begin{shaded}
%We define iterated composed functions, for function $f:\boolfn{n}$, let $f^{\otimes k}:\boolfn{n^k}$ be recursively defined as $f\circ f^{\otimes (k-1)}$ where $f^{\otimes 1}=f$ will be called the base function. The iterated composed functions for the base function $\AND_2\circ \OR_2$ and $\Maj_3$ will interest us later.
For any function $f:\boolfn{n}$, that we will call the base function, the iterated composed function  $f^{\otimes k}:\boolfn{n^k}$ is recursively defined as   $f^{\otimes 1}=f$, $f^{\otimes k}=f\circ f^{\otimes (k-1)}$. The iterated composed functions for the base functions $\AND_2\circ \OR_2$ and $\Maj_3$ will interest us later.
%\end{shaded}

\paragraph*{Symmetric functions}
A Boolean function is symmetric if its value depends only on the number of ones in the input, and not on the positions of the ones. 
\begin{proposition}\label{prop:symm_evasive}
  For every non-constant symmetric boolean function $f: \boolfn{n}$,
  \begin{enumerate}
  \item $f$ is evasive (has $\Depth(f)=n$). (See eg.\ Lemma 14.19
    \cite{Jukna-BFCbook}.)
  \item Hence, for any weights $w_i$, $\wDepth(f,[w_1,w_2,...,w_n] ))=\sum_i w_i$.
  \end{enumerate}
\end{proposition}

For a symmetric Boolean function $f: \boolfn{n}$,
let $f_0,f_1,...,f_n\in \bool$ denote the values of the function $f$ on inputs of Hamming weight $0,1,...,n$ respectively. The $\Gap$ of $f$ is defined as the length of the longest interval (minus one) where $f_i$  is constant. That is, 
\[\Gap(f) = \max_{0\leq a \leq b \leq n} \{b-a: f_a=f_{a+1}=...=f_b \}.\]
Analogously, $\minGap(f)$ is the length of the shortest constant interval (minus one); that is, setting $f_{-1}\neq f_0$ and $f_{n+1}\neq f_{n}$ for boundary conditions, 
\[\minGap(f)=\min_{0\leq a \leq b \leq n} \{b-a: f_{a-1}\neq f_a=f_{a+1}=...=f_b\neq f_{b+1} \}.\] 

\paragraph*{Fourier Representation of Boolean functions}
We include here some basic facts about Fourier representation relevant to our work. For a wonderful comprehensive overview of this area, see \cite{o2014analysis}. 
Consider the inner product space of functions $\mathcal{V}=\{f: \boolset{n}\longrightarrow \mathbb{R}\}$ with the inner product defined as 
\[
\langle f,g \rangle = \frac{1}{2^n} \sum_{x\in \boolset{n}}f(x)g(x).
\]

For $S\subseteq [n]$, 
the function $\chi_{S}: \boolset{n}\longrightarrow \{-1,1\}$ defined by $\chi_{S}(x)=(-1)^{\sum_{i\in S}x_i}$ is the $\pm 1$ parity of the bits in $S$ and therefore is referred to as a parity function. The set of all parity functions $\{\chi_S: S\subseteq[n]\}$ forms an orthonormal basis for $\mathcal{V}$. Thus, every function $f\in \mathcal{V}$, in particular boolean functions, has a unique representation $f=\sum_{S\subseteq[n]}\hat{f}(S)\chi_{S}$. The coefficients $\{\hat{f}(S): S\subseteq [n]\}$ are called the Fourier coefficients(spectrum) of $f$.
The Fourier sparsity of $f$, denoted by $\spar(f)$, is the number of non-zero Fourier coefficients in the expansion of $f$, i.e. $\lvert \{S \subseteq [n]: \hat{f}(S)\neq 0 \}\rvert$. It will be convenient for us to disregard the Fourier coefficient of the empty set. We therefore define $\sparn(f) = \lvert \{S \subseteq [n]: S\neq \emptyset; \hat{f}(S)\neq 0 \}\rvert$. For every $f$, $0\le \sparn(f) \le \spar(f) \le \sparn(f)+1$, and  only the constant functions have $\sparn=0$. 

%For a Boolean function $f: \boolfn{n}$, consider the (essentially equivalent) function $f':\boolset{n} \longrightarrow \{-1,+1\}$ where $f'(x) = 1-2f(x)$.  It is convenient to think of $f'$ when considering spectral parameters. To keep notation unambiguous, we use $\sparh(f)$ to  denote the sparsity of $f'$;
%  $\sparh(f) = \spar(1-2f)$. Note that $|\spar(f)- \sparh(f)|\in \{0,1\}$ (since only the coefficient for $S=\emptyset$ can flip between zero and non-zero), and that the constant functions, and the parity functions, have sparsity $\sparh= 1$; all other Boolean functions have $\sparh$ greater than 1. 

%    The number of non-zero Fourier coefficients, i.e. $\lvert \{S \subseteq [n]: \hat{f}(S)\neq 0 \}\rvert$ is the Fourier sparsity of $f$, denoted by $\spar(f)$. 
%The number of non-zero Fourier coefficients, excluding $S=\emptyset$, i.e. $\lvert \{S:S\neq \emptyset, \hat{f}(S)\neq 0 \}\rvert$ is the Fourier sparsity of $f$, denoted by $\spar(f)$. We exclude the $S=\emptyset$ from our definition to simplify the presentation.

Sparsity is related to decision tree complexity; large sparsity implies large depth.

%Next, we state some relation between the Fourier spectrum and the decision tree complexity. See \cite{o2014analysis} for a wonderful overview of this area.
%First, we note a relation between sparsity and the depth of a function, which says that the sparsity being large implies large depth.
\begin{proposition}[see Proposition 3.16 in \cite{o2014analysis}]\label{prop:depth-sparsity}
For a Boolean function $f:\boolfn{n}$, $\log \spar(f)\le \log \Size(f)+ \Depth(f)\le 2\Depth(f)$. 
\end{proposition}

In our discussion, we will be interested in the effect of restrictions on the Fourier representation of a function. Of particular interest to us will be restrictions to subcubes. A subcube is a set of all inputs consistent with a partial assignment of $n$ bits. Formally, a subcube $J$ is a partial assignment (to some of the $n$ variables) defined by $(S,\rho)$ where $S\subseteq [n]$ is the set of input bits fixed by $J$ and $\rho: S \longrightarrow \bool$ is the map according to which the bits in $S$ are fixed. A subcube is a special type of affine subspace; hence, inheriting notation from subspaces, for $J=(S,\rho)$, the cardinality of $S$ is called the co-dimension of $J$, and is denoted by $\codim(J)$. A function $f:\boolfn{n}$ restricted to $J=(S,\rho)$ is the function $f|J:\bool^{n- |S|} \longrightarrow \bool$ obtained by fixing variables in $S$ according to $\rho$. 
The following result quantifies the effect on Fourier spectrum of subcube restriction. % of querying fixed variables from a subcube on which the function becomes constant.
\begin{restatable}{theorem}{thmrestsparsity}[\cite{ShpilkaTalVolk17,MandeSanyal-FSTTCS20}]
  \label{thm:rest-sparsity}
  %\begin{theorem}[\cite{ShpilkaTalVolk17,MandeSanyal-FSTTCS20}]\label{thm:rest-sparsity}
  Let $f$ be any Boolean function  $f:\boolfn{n}$. 
   Fix any $S \subseteq [n]$, $S \neq \emptyset$.  If $f|(S,\rho)$ is a constant, then for every $\rho': S \longrightarrow \bool$,
 $\sparn(f|(S,\rho')) \le \sparn(f)/2$.  
% Let $f$ be constant on subcube $J=(S,\rho)$. Then for any $\rho': S \longrightarrow \bool$, $f$ restricted to $J'=(S,\rho')$ has sparsity at most $\spar(f)/2$ i.e. $\spar(f|J')\le \spar(f)/2$.
%\end{theorem}
\end{restatable}

%\begin{shaded}
This lemma follows from \cite{ShpilkaTalVolk17} (in the proof of Theorem 1.7 there) and \cite{MandeSanyal-FSTTCS20} (see the discussion in Sections 2.1 and 3.1 there). Both papers consider affine subspaces, of which subcubes are a special case. Since  the result is not explicitly stated in this form in either paper, for completeness we give a proof  for the subcubes case in the appendix.
%\end{shaded}

The subcube kill number of $f$, denoted by $\skill(f)$, measures a largest subcube over which $f$ is constant, and is defined as 
\[
\skill(f)=\min \{\codim(J)| f|J \text{ is constant} \}.
\]

\paragraph*{Decision Tree Rank}
For a rooted binary tree $T$, the rank of the tree is the rank of the root
node, where the rank of each node of the tree is defined recursively
as follows:
For a leaf node $u$, $\Rank(u)=0$.
For an internal node $u$  with children $v,w$,
\[
\Rank(u) = \left\{
\begin{array}{ll} \Rank(v) + 1 & \textrm{~~if~} \Rank(v)=\Rank(w) \\
\max\{\Rank(v),\Rank(w)\} & \textrm{~~if~} \Rank(v)\neq \Rank(w) \\
\end{array}
\right.
\]
The following proposition lists some known properties of the rank function
for binary trees. 
\begin{proposition}\label{prop:prop_rank_tree}
  For any binary tree $T$, 
\begin{enumerate}
\item \label{item-rank-size} (Rank and Size relationship): $\Rank(T) \le \log(\Size(T)) \le \Depth(T)$. 
\item \label{item-monotonicity} (Monotonicity of the Rank): Let $T'$ be any subtree of $T$, and let $T''$ be an arbitrary binary tree of higher rank than $T'$. If  $T'$ is replaced by $T''$ in $T$, then the rank of the resulting tree is not less than the rank of $T$.
\item \label{item-leaf-depth-rank} (Leaf Depth and Rank): If all leaves in $T$ have depth at least $r$, then $\Rank(T)\ge r$.
\end{enumerate}   
\end{proposition}

For a Boolean function $f$, the rank of $f$, denoted $\Rank(f)$, is
the minimum rank of a decision tree computing $f$. 

From \cref{prop:prop_rank_tree}(\ref{item-monotonicity}), we see that
the rank of a subfunction of $f$ (a function obtained by assigning
values to some variables of $f$) cannot exceed the rank of the function itself.
\begin{proposition}\label{prop:rank_subfn}
(Rank of a subfunction): Let $f_S$ be a subfunction obtained by fixing the values of variables in some set $S\subseteq [n]$ of $f$. Then $\Rank(f_S) \le \Rank(f)$. 
\end{proposition}
The following rank and size relationship is known for boolean functions.
\begin{proposition}[Lemma 1 \cite{EH-IC1989}]\label{prop:rank_size}
  For a non-constant Boolean function $f: \boolfn{n}$,
$$\Rank(f)\le \log \Size(f) \le \Rank(f)\log \left(\frac{e n}{\Rank(f)}\right).$$  
\end{proposition}
It follows that $\Rank(f) \in \Theta(\log \Size(f))$ if and only if $\Rank(f) = \Omega(n)$. However, even when $\Rank(f)\in o(n)$, it characterizes $\log\Size(f)$ upto a logarithmic factor, since for every $f$, $\Rank(f) \in \Omega(\log\Size(f)/\log n)$.

For symmetric functions,  $\Rank$ is completely characterized in terms of $\Gap$. 
\begin{proposition}[Lemma C.6 \cite{ABDORU10}]\label{lem:ABDORU}
  For symmetric Boolean function $f: \boolfn{n}$, $\Rank(f) = n - \Gap(f)$.
\end{proposition}

\begin{remark}\label{rem:measures-neg-dual}
For (simple) deterministic possibly weighted decision trees, each of the measures \Size, \Depth, and \Rank, is the same for a Boolean function $f$, its complement $\neg f$, and its dual $f^d$. 
\end{remark}

\paragraph*{Relations and Search problems}
A relation $R \subseteq X \times W$ is said to be $X$-complete, or just complete, if its
projection to $X$ equals $X$. That is, for every $x\in X$, there is a
$w\in W$ with $(x,w)\in R$. For an $X$-complete relation $R$, where
$X$ is of the form $\boolset{n}$ for some $n$, the search problem
$\SearchR$ is as follows: given an $x\in X$, find  a $w\in W$ with
$(x,w)\in R$. A decision tree for $\SearchR$ is defined exactly as for
Boolean functions; the only diference is that leaves are labeled with
elements of $W$, and we require that for each input $x$, if the unique
leaf reached on $x$ is labeled $w$, then $(x,w) \in R$. The rank of
the relation, $\Rank(R)$, is the minimum rank of a decision tree
solving the $\SearchR$ problem.

A Boolean function $f:\boolfn{n}$ naturally defines a complete
relation $R_f$ over $X=\boolset{n}$ and $W=\bool$, with $R_f =
\{(x,f(x)) \mid x\in X \}$, and $\Rank(f) = \Rank(R_f)$.

\section{Game Characterisation for Rank}
\label{sec:game}

In this section we observe that the rank of a Boolean function is
characterised by the value of a Prover-Delayer game introduced by
Pudl{\'a}k and Impagliazzo in \cite{pudlak2000lower}. As mentioned in
\cref{sec:intro}, the game was originally described for searchCNF
problems on unsatsifiable clause sets. The appropriate analog for a
Boolean function $f$, or its relation $R_f$, and even for arbitrary
$X$-complete relations $R\subseteq X\times W$, is as follows:

The game is played by two players, the Prover and the Delayer, who
construct a (partial) assignment $\rho$ in rounds.  Initially, $\rho$
is empty. In each round, the Prover queries a variable $x_i$ not set
by $\rho$. The Delayer responds with a bit value $0$ or $1$ for $x_i$,
or defers the choice to the Prover. In the later case, Prover can
choose the value for the queried variable, and the Delayer scores one
point. The game ends when there is a $w\in W$ such that for all $x$
consistent with $\rho$, $(x,w)\in R$. (Thus, for a Boolean function
$f$, the game ends when $f\vert_\rho$ is a constant function.)  The
value of the game, $\Value(R)$, is the maximum $k$ such that the
Delayer can always score at least $k$ points, no matter how the Prover
plays.
\begin{theorem}[implied from \cite{pudlak2000lower,Kullmann-ECCC-99,esteban2003combinatorial}]
  \label{thm:game-rank}
  For any $X$-complete relation $R \subseteq X \times W$, where $X =
  \boolset{n}$, $\Rank(R) = \Value(R)$.  In particular, 
  for a boolean function $f: \boolfn{n}$, $\Rank(f) =
  \Value(R_f)$. 
\end{theorem}

The proof of the theorem follows from the next two lemmas.
\begin{lemma}[implicit in \cite{Kullmann-ECCC-99}]\label{lem:game-rank-ub}
For an $X$-complete relation $R \subseteq \boolset{n}\times W$, in the
Prover-Delayer game, the Prover has a strategy which restricts the
Delayer's score to at most $\Rank(R)$ points.
\end{lemma}   
\begin{proof}
  The Prover chooses a decision tree $T$ for $\SearchR$ and starts querying variables starting from the root and working down the tree. If the Delayer responds with a $0$ or a $1$, the Prover descends into the left or right subtree respectively. If the Delayer defers the decision to Prover, then the Prover sets the variable to that value  for which the corresponding subtree has smaller rank (breaking ties arbitrarily), and descends into that subtree. 

 We claim that such a ``tree-based'' strategy restricts the Delayer's
 score to $\Rank(T)$ points.  The proof is by induction on
 $\Depth(T)$.
\begin{enumerate}
\item Base Case: $\Depth(T)=0$. This means that $\exists w\in W$,
  $X\times \{w\}\subseteq R$. Hence the game terminates with the empty
  assignment and the Delayer scores 0.
\item Induction Step: $\Depth(T)\ge 1$. Let $x_i$ be the variable at
  the root node and $T_0$ and $T_1$ be the left and right subtree. The
  Prover queries the variable $x_i$. Note that for all $b$,
  $\Depth(T_b) \le \Depth(T)-1$, and $T_b$ is a decision tree for the
  search problem on $R_{i,b}\triangleq \{ (x,w)\in R \mid x_i=b\}
  \subseteq X_{i,b} \times W$, where $X_{i,b} = \{x\in X\mid x_i=b\}$.

  If the Delayer responds with a bit $b$, then by induction, the
  subsequent score of the Delayer is limited to $\Rank(T_b) \le
  \Rank(T)$.  Since the current round does not increase the score, the
  overall Delayer score is limited to $\Rank(T)$.

  If the Delayer defers the decision to Prover, the Delayer gets one
  point in the current round. Subsequently, by induction, the
  Delayer's score is limited to $\min(\Rank(T_0), \Rank(T_1))$; by
  definition of rank, this is at most $\Rank(T)-1$. So the overall
  Delayer score is again limited to $\Rank(T)$.
\end{enumerate}
In particular, if the Prover chooses a rank-optimal tree $T_R$, then the Delayer's score is limited to $\Rank(T_R) = \Rank(R)$ as claimed.
\end{proof}

\begin{lemma}[implicit in \cite{esteban2003combinatorial}]\label{lem:game-rank-lb}
For an $X$-complete relation $R \subseteq \boolset{n}\times W$, in the
Prover-Delayer game,
the Delayer has a strategy which always scores at least $\Rank(R)$ points.
\end{lemma}
\begin{proof}
The Delayer strategy is as follows: When variable $x_i$ is queried,
the Delayer responds with $b\in\bool$ if $\Rank(R_{i,b}) >
\Rank(R_{i,1-b})$, and otherwise defers.

We show that the Delayer can always score $\Rank(R)$ points using this
strategy. The proof is by induction on the number of variables $n$.
Note that if $\Rank(R)=0$, then
there is nothing to prove. If $\Rank(R)\ge1$, then the prover must
query at least one variable. 
\begin{enumerate}
\item Base Case: $n=1$. If $\Rank(R)=1$, then the prover must query
  the variable, and the Delayer strategy defers the choice, scoring
  one point.
\item Induction Step: $n>1$. Let $x_i$ be first variable queried by
  the prover.

  If $\Rank(R_{i,0}) = \Rank(R_{i,1})$, then the Delayer defers,
  scoring one point in this round. Subsequently, suppose the Prover
  sets $x_i$ to $b$. The game is now played on $R_{i,b}$, and by
  induction, the Delayer can subsequently score at least
  $\Rank(R_{i,b})$ points. But also, because of the equality, we have
  $\Rank(R) \le 1+\Rank(R_{i,b})$, as witnessed by a decision tree that
  first queries $x_i$ and then uses rank-optimal trees on each branch.
  Hence the overall Delayer score is at least $\Rank(R)$. 

  If $\Rank(R_{i,b}) > \Rank(R_{i,1-b})$, then the Delayer chooses
  $x_i=b$ and the subsequent game is played on $R_{i,b}$. The
  subsequent (and hence overall) score is, by induction, at least
  $\Rank(R_{i,b})$. But $\Rank(R) \le \Rank(R_{i,b})$, as witnessed by
  a decision tree that first queries $x_i$ and then uses rank-optimal
  trees on each branch. 
\end{enumerate}
\end{proof}

\cref{lem:game-rank-ub,lem:game-rank-lb} give us a way to prove rank upper and lower bounds for boolean functions. In a Prover-Delayer game for $R_f$, exhibiting a Prover strategy which restricts the Delayer to at most $r$ points gives an upper bound of $r$ on  $\Rank(f)$. Similarly, exhibiting a Delayer strategy which scores at least $r$ points irrespective of the Prover strategy shows a lower bound of $r$ on $\Rank(f)$.

In \cite{BGL13}, an aysmmmetric version of
this game is defined. In each round, the Prover queries a variable
$x$, the Delayer specifies values $p_0,p_1 \in [0,1]$ adding up to 1, 
%(a probability distribution on $\bool$),
the Prover picks a value $b$,
the Delayer adds $\log \frac{1}{p_b}$ to his score. Let $\AValue$
denote the maximum score the Delayer can always achieve, independent
of the Prover moves. Note that $\AValue(R) \ge \Value(R)$; an
asymmetric-game Delayer can mimic a symmetric-game Delayer by using
$p_b=1$ for choice $b$ and $p_0=p_1=1/2$ for deferring. As shown in
\cite{BGL13}, for the search CNF problem, the value of this asymmetric
game is exactly the optimal leaf-size of a decision tree. We note below that
this holds for the $\SearchR$ problem more generally.
\begin{proposition}[implicit in \cite{BGL13}]
  \label{prop:game-size}
  For any $X$-complete relation $R \subseteq X \times W$, where $X =
  \boolset{n}$, $\log\Size(R) = \AValue(R)$.  In particular, 
  for a boolean function $f: \boolfn{n}$, $\log\Size(f) =
  \AValue(R_f)$. 
\end{proposition}
(In \cite{BGL13}, the bounds have $\log (S/2)$; this is because $S$
there counts all nodes in the decision tree, while here we count only leaves.)

Thus we have the relationship 
\[\Rank(f) = \Value(R_f) \le \AValue(R_f) = \log \Size(f).\]
This relationship explains where the slack may lie in the inequalities
from \cref{prop:rank_size} relating $\Rank(f)$ and $\log \Size
(f)$. The symmetric game focusses on the more complex subtree,
ignoring the contribution from the less complex subtree (unless both
are equally complex), and thus characterizes rank. The asymmetric game
takes a weighted contribution of both subtrees and thus is able to
characterize size.

\section{The Rank of some natural functions}
\label{sec:simple-calc}
For symmetric functions, rank can be easily calculated using
\cref{lem:ABDORU}. In \cref{tab:tabulation} we tabulate various
measures for some standard symmetric functions.
As can be seen from the $\OR_n$ and $\AND_n$ functions, the $\Rank(f)$
measure is not polynomially related with the measures $\Depth(f) $ or
certificate complexity $\Cert(f)$.
\begin{table}[h]
\[\begin{array}{|c|c|c|c|c|c|c|}
\hline
  $f$ & \Depth & \CertZ & \CertO & \Cert & \Gap & \Rank \\ \hline
  0 ~\textrm{or}~ 1 & 0 & 0 & 0 & 0 & n & 0 \\ \hline
  \AND_n & n & 1 & n & n & n-1 & 1 \\ \hline
  \OR_n & n & n & 1 & n & n-1 & 1 \\ \hline
  \Parity_n & n & n & n & n & 0 & n \\ \hline
  \Maj_{2k} & 2k & k & k+1 & k+1 & k & k \\ \hline
  \Maj_{2k+1} & 2k+1 & k+1 & k+1 & k+1 & k & k+1 \\ \hline  
%  \Thr^k_n ~(k \ge 1) & n & n-k+1 & k & \max\{n-k+1,k\} & \max\{ k-1, n-k\} & n-\Gap \\ \hline
%  \begin{array}{c}\Thr^k_n \\(k \ge 1) \end{array} &
%  n & n-k+1 & k &
%  \max\left\{\begin{array}{c}n-k+1,\\k\end{array}\right\} &
%  \max\left\{ \begin{array}{c}k-1,\\ n-k \end{array}\right\} & n-\Gap
  \begin{array}{c}\Thr^k_n \\(k \ge 1) \end{array} &
  n & n-k+1 & k &
  \max \begin{Bmatrix}n-k+1,\\k\end{Bmatrix} &
  \max\begin{Bmatrix}k-1,\\ n-k \end{Bmatrix} & n-\Gap
  \\ \hline
\end{array}\]
  \caption{Some simple symmetric functions and their associated complexity measures}
  \label{tab:tabulation}
\end{table}

For two composed functions  that will be crucial in our discussions in \cref{sec:rank-rels}, we can directly calculate the rank as described  below. (The rank can also be caluclated using \cref{thm:game-rank}; see \cref{sec:game-proofs}, or using \cref{thm:compose-rank-bounds}, which is much more general. We show these specific bounds here since we use them in \cref{sec:rank-rels}.)
\begin{theorem}\label{thm:rank-tribes}
  For every $n\ge 1$,
  \begin{enumerate}
  \item $\Rank(\Tribes_{n,m}) = \Rank(\dTribes_{n,m}) = n$ for $m\ge 2$.
  \item $\Rank(\AND_n \circ \Parity_m) = n(m-1) +1$ for $m\ge 1$.
  \end{enumerate}
\end{theorem}

We prove this theorem by proving each of the lower and upper bounds separately in a series of lemmas below. The lemmas use the following property about the rank function.
\begin{proposition}\label{prop:compose_rank_dt}
  (Composition of  Rank):
  Let $T$ be a rooted binary tree with depth $\ge 1$, rank $r$,  and with leaves labelled by $0$ and $1$.
  Let $T_0, T_1$ be arbitrary rooted binary trees of ranks $r_0, r_1$ respectively.
  For $b\in\bool$, attach $T_b$ to each leaf of $T$ labeled $b$, to obtain rooted binary tree $T'$ of rank $r'$.
  \begin{enumerate}
    \item $r' \le r + \max\{r_0,r_1\}$.  Furthermore, if $T$ is a
      complete binary tree, and if $r_0=r_1$, then this is an
      equality; $r'=r+r_0$.
    \item If every non-trivial subtree (more than one leaf) of $T$ has
      both a $0$ leaf and a $1$ leaf, then $r'\ge r + \max\{r_0,r_1\}
      -1$.   If, furthermore,  $T$ is a complete binary tree, then
  this is an equality when $r_0\neq r_1$, 
  \end{enumerate} 
\end{proposition}

\begin{proof}
The upper bound on $r'$ follows from the definition of rank when $r_0=r_1$, in which case it also gives equality for complete $T$. When $r_0\neq r_1$, it follows from  \cref{prop:prop_rank_tree}(\ref{item-monotonicity}).

For non-trivially labeled $T$, we establish the lower bound by induction on $d=\Depth(T)$.

In the base case $d=1$, $T$ has one 0-leaf and one 1-leaf, and $r=1$. By definition of rank, $r'$ satisfies the claimed inequality. 

For the inductive step, let $\Depth(T)= k > 1$. Let $v$  be the root of $T$, and let $T_\ell$, $T_r$ be its left and right sub-trees respectively, with ranks $r_\ell$ and $r_r$ respectively. Both $\Depth(T_\ell)$ and $\Depth(T_r)$ are at most $k-1$, and at least one of these is exactly $k-1 \ge 1$. Also, at least one of $r_\ell, r_r$ is non-zero.

Let $T'_\ell$  be the tree obtained by replacing 0 and 1 leaves of $T_\ell$ by $T_0$ and $T_1$ respectively; let its rank be $r'_\ell$. Similarly construct $T'_r$, with rank $r'_r$. Then $T'$ has root $v$ with left and right subtrees $T'_\ell$ and $T'_r$.

If $r_\ell=0$, then $r=r_r$ and $\Depth(T_r) = k-1\ge 1$. By the induction hypothesis, $r_r + \max\{r_0,r_1\} -1 \le r'_r$. Since $r' \ge r'_r$, the claimed bound follows.

If $r_r=0$, a symmetric argument applies.

If both $r_\ell, r_r$ are positive, then by the induction hypothesis,
$r_\ell + \max\{r_0,r_1\} -1 \le r'_\ell$ and $r_r + \max\{r_0,r_1\}
-1 \le r'_r$. If $r_\ell=r_r$ then $r=r_\ell+1$, and by definition of rank, 
$r' \ge 1 + \min\{r'_\ell, r'_r\} \ge r_\ell + \max\{r_0,r_1\}  = 
r+ \max\{r_0,r_1\} -1$, as claimed. On the other hand, if $r_\ell\neq r_r$, then
$r = \max\{r_\ell,r_r\}$, and by definition of rank, 
$r' \ge \max\{r'_\ell, r'_r\} \ge \max\{r_\ell,r_r\} + \max\{r_0,r_1\} -1 =
r+ \max\{r_0,r_1\} -1$, as claimed.

For complete binary tree $T$ satisfying the labelling requirements, 
$r_\ell=r_r=r-1$. 
The same arguments, simplified to this situation, show the claimed equality:
\end{proof}

We first establish the bounds for $\dTribes_{n,m} = \bigwedge_{i\in [n]}
\bigvee_{j\in [m]} x_{i,j}$.
\begin{lemma}\label{lem:rank-tribes-ub}
  For every $n,m \ge 1$, $\Rank(\dTribes_{n,m}) \le n$.
\end{lemma}
\begin{proof}
  We show the bound by giving a recursive construction and bounding the rank by induction on $n$. 
In the base case, $n=1$. $\dTribes_{1,m} = \OR_m$, which has rank 1.
For the inductive step, $n> 1$. For $j\le n$, let $T_{j,m}$ denote the recursively constructed trees for $\dTribes_{j,m}$. Take the tree $T$ which is $T_{1,m}$ on variables $x_{n,j}$, $j\in[m]$. Attach the tree $T_{n-1,m}$ on variables $x_{i,j}$ for $i\in[n-1]$, $j\in[m]$, to all the 1-leaves of $T$, to obtain $T_{n,m}$. It is straightforward to see that this tree computes $\dTribes_{n,m}$. Using \cref{prop:compose_rank_dt} and induction, we obtain  $\Rank(T_{n,m}) \le \Rank(T_{1,m}) + \Rank(T_{n-1,m}) \le 1 + (n-1) = n$. 

\end{proof}

\begin{remark}\label{rem:and-composed-f-ub}
  More generally, this construction shows  that $\Rank(\AND_n\circ f) \le n\Rank(f)$.  
\end{remark}

\begin{lemma}\label{lem:rank-tribes-lb}
  For every $n\ge 1$ and $m \ge 2$, $\Rank(\dTribes_{n,m}) \ge n$.
\end{lemma}
\begin{proof}
  We prove this by induction on $n$. The base case, $n=1$, is straightforward: $\dTribes_{1,m}$ is the function $\OR_m$, whose rank is $1$. 

  For the inductive step, let $n> 1$, and consider any decision tree $Q$ for $\dTribes_{n,m}$. Without loss of generality (by renaming variables if necessary), let $x_{1,1}$ be the variable queried at the root node. Let $Q_0$ and $Q_1$ be the left and the right subtrees of $Q$. Then $Q_0$ computes the function $\AND_n\circ(\OR_{m-1}, \OR_{m},..., \OR_{m})$, and $Q_1$ computes $\dTribes_{n-1,m}$, on appropriate variables. For $m\ge 2$, $\dTribes_{n-1,m}$ is a sub-function of $\AND_n\circ(\OR_{m-1}, \OR_{m},..., \OR_{m})$, and so  \cref{prop:rank_subfn} implies that  $\Rank(Q_0)\ge \Rank(\AND_n\circ(\OR_{m-1}, \OR_{m},..., \OR_{m})) \ge \Rank(\dTribes_{n-1,m})$. By induction, $\Rank(Q_1) \ge \Rank(\dTribes_{n-1.m})\ge n-1$. Hence, by definition of rank, $\Rank(Q) \ge 1+\min\{\Rank(Q_0),\Rank(Q_1\} \ge n$. Since this holds for every decision tree $Q$ for $\dTribes_{n,m}$, we conclude  that $\Rank(\dTribes_{n,m})\ge n$, as claimed.
\end{proof}

Next, we establish the bounds for $\AND_n \circ \Parity_m = \bigwedge_{i\in [n]} \bigoplus_{j\in [m]} x_{i,j}$. The upper bound below is slightly better than what is implied by \cref{rem:and-composed-f-ub}.
\begin{lemma}\label{lem:and-parity-ub}
For every $n,m \ge 1$, $\Rank(\AND_n \circ \Parity_m) \le n(m-1) +1$.
\end{lemma}
\begin{proof}
  Recursing on $n$, we construct decision trees $T_{n,m}$ for $\AND_n \circ \Parity_m$, as in \cref{lem:rank-tribes-ub}. By induction on $n$, we bound the rank, also additionally using the fact that the rank-optimal decision tree for $\Parity_m$ is a complete binary tree.
  
  Base Case: $n=1$. $\AND_1 \circ \Parity_m = \Parity_m$. From \cref{tab:tabulation}, $\Rank(\Parity_m)=m$; let $T_{1,m}$ be the optimal decision tree computing $\Parity_m$.

  Inductive Step: $n> 1$. For $j\le n$, let $T_{j,m}$ denote the recursively constructed trees for $\AND_j\circ\Parity_{m}$. Take the tree $T$ which is $T_{1,m}$ on variables $x_{n,j}$, $j\in[m]$. Attach the tree $T_{n-1,m}$ on variables $x_{i,j}$ for $i\in[n-1]$, $j\in[m]$, to all the 1-leaves of $T$, to obtain $T_{n,m}$. It is straightforward to see that this tree computes $\AND_n\circ\Parity_{m}$.

  By induction, $\Rank(T_{n-1,m}) \le (n-1)(m-1)+1\ge 1$. Since we do not attach anything to the 0-leaves of $T_{1,m}$ (or equivalently, we attach a rank-0 tree to these leaves), and since $T_{1,m}$ is a complete binary tree, the second statement in \cref{prop:compose_rank_dt} yields  $\Rank(T_{n,m}) = \Rank(T_{1,m}) + \Rank(T_{n-1,m}) -1$. Hence $\Rank(T_{n,m}) \le n(m-1)+1$, as claimed.
\end{proof}

\begin{lemma}\label{lem:and-parity-lb}
  For every $n, m_1, m_2, \ldots , m_n\ge 1$, and functions $g_1, g_2, \ldots , g_n$  each in $\{\Parity_m, \neg\Parity_m\}$,
  $\Rank(\AND_n \circ (g_1,g_2,...,g_n)) \ge (\sum_{i=1}^n (m_i-1))+1$.

  In particular, $\Rank(\AND_n\circ \Parity_m)\ge n(m-1)+1$.
\end{lemma}
\begin{proof}
  We proceed by induction on $n$.   Let $h$ be the function $\AND_n \circ (g_1,g_2,...,g_n)$.

Base Case: $n=1$.  $h = g_1$.  Note that for all functions $f$, $\Rank(f)=\Rank(\neg~f)$. So $\Rank(h)=\Rank(\Parity_{m_1}) = m_1$. 
Inductive Step: $n>1$. We proceed by induction on $M=\sum_{i=1}^n m_i$. 
\begin{enumerate}
\item Base Case: $M=n$. Each $m_i$ is equal to $1$. So $h$ is the conjunction of $n$ literals on distinct variables. (A literal is a variable or its negation.) Hence $\Rank(h) = \Rank(\AND_n) = 1$. 
\item Inductive Step: $M>n>1$.
  Consider any decision tree $Q$ computing $h$.  Without loss of generality (by renaming variables if necessary), let $x_{1,1}$ be the variable queried at the root node. Let $Q_0$ and $Q_1$ be the left and the right subtrees of $Q$. For $b\in\bool$, let $g_{1b}$ denote the function $g_1$ restricted to $x_{1,1}=b$. Then $Q_b$ computes the function $\AND_n\circ(g_{1b},g_2, \ldots , g_n)$ on appropriate variables.
  
  If $m_1=1$,  then the functions  $g_{10},g_{11}$ are constant functions, one 0 and the other 1. So one of $Q_0,Q_1$ is a 0-leaf, and the other subtree computes $\AND_{n-1}\circ (g_2,...,g_n)$. Using induction on $n$, we conclude
\[\Rank(Q) \ge \Rank(\AND_{n-1}\circ (g_2,...,g_n)) \ge \left[\sum_{i=2}^n(m_i-1)\right]+1 = \left[\sum_{i=1}^n(m_i-1)\right]+1. \]

For $m_1\ge 2$, $\{g_{10},g_{11}\}=\{\Parity_{m_1-1},\neg \Parity_{m_1-1}\}$. 
So one of $Q_0,Q_1$ computes
$\AND_{n}\circ (\Parity_{m_1-1},g_2,...,g_n)$, and the other computes $\AND_{n}\circ (\neg \Parity_{m_1-1},g_2,...,g_n)$. Using induction on $M$, we obtain
\[
%\Rank(Q) \ge  1 + \min\{\Rank(Q_0),\Rank(Q_1)\}
\Rank(Q) \ge  1 + \min_b \Rank(Q_b)
\ge 1 + (m_1-2)+\left[\sum_{i=2}^n (m_i-1)\right] + 1
= \left[\sum_{i=1}^n (m_i-1)\right] + 1. 
\]

Since this holds for every decision tree $Q$ for $h$, the induction step is proved. 
\end{enumerate}
\end{proof}

\section{Relation between Rank and other measures}
\label{sec:rank-rels}

%\begin{shaded}
\subsection{Relating Rank to Depth and Sparsity}
\label{sec:rank-depth}
From \cref{prop:prop_rank_tree,prop:rank_size}, we know that $\Depth(f)$ is at least $\Rank(f)$. In the other direction, the $\AND$ function with rank $1$ and depth $n$ shows that $\Depth(f)$ cannot be bounded from above by any function of $\Rank(f)$ alone.
Similarly, we know from \cref{prop:depth-sparsity} that $\Depth(f)$ is bounded from below by $\log\spar(f)/2$, and yet, as witnessed by the $\Parity$ function with depth $n$ and sparsity 1, it cannot be bounded from above by any function of $\log\spar(f)$ alone. 
We show in this section that a combination of these two measures  does bound $\Depth(f)$ from above.  Thus, in analogy to  \cref{prop:depth-sparsity,prop:rank_size}, we see where  $\Depth(f)$ is sandwiched:
\[ \max\{\Rank(f),\log\spar(f)/2\} \le \Depth(f) \le \Rank(f) (1+\log\spar(f))\]

To establish the upper bound, we first observe that subcube kill number is bounded above by rank.
\begin{lemma}\label{lemma:kill}
For every Boolean function $f$, $\skill(f) \le \max_{\text{subcube } J } \skill(f|J) \le \Rank(f)$. 
\end{lemma} 
\begin{proof}
The first inequality holds since $f=f|J$ for the unique subcube $J$ of codimension 0.

Next, note that showing $\skill(g)\le \Rank(g)$ for every boolean function $g$ suffices to prove the lemma. This is because if we prove this, then for every subcube $J$, $\skill(f|J) \le \Rank(f|J) \le \Rank(f)$; the latter inequality
follows  by monotonicity of rank (\cref{prop:rank_subfn}). 
%As assuming $\skill(f)\le \Rank(f)$, we have that $\max_{\text{subcube } J } \skill(f|J) \le \max_{\text{subcube } J } \Rank(f|J)$ and by monotonicity of rank \cref{prop:prop_rank_tree}(\ref{item-monotonicity}), for any $J$,  $\Rank(f|J) \le \Rank(f)$.\\

To show $\skill(g)\le \Rank(g)$,  observe the following property of rank, which follows from the definition: For every internal node $v$ in a tree, at least one of its children has rank strictly less than the rank of $v$. 
%The rank of any internal node in a tree is at least one more than the minimum rank among the two children nodes.
Now, let $T$ be a rank-optimal tree for $g$, of rank $r$. Starting from the root in $T$, traverse down the tree in the direction of smaller rank until a leaf is reached. Using the above property, we see that we reach a leaf node in $T$ at  depth at most  $r$.  The variables queried on the path leading to the leaf node, and their settings consistent with the path, give a subcube of co-dimension at most $r$. On this subcube, since a decision tree leaf is reached, $g$ becomes constant, proving the claim.
\end{proof}

Combining the \cref{lemma:kill} and \cref{thm:rest-sparsity}, we show the following.
\begin{theorem}
\label{thm:rank-sparsity-depth}
  For every Boolean function $f: \boolfn{n}$,
%$$ \Depth(f) \le \Rank(f)(1 + \log (\sparh(f))) \le \Rank(f)(1 + \log (1+\spar(f))).$$
$$ \Depth(f) \le \Rank(f)(1 + \log (\spar(f))).$$
  The inequality is tight as witnessed by $\Parity$ function.
\end{theorem} 
\begin{proof}
  %% Let $f'$ be the function $\{0,1\}^n \longrightarrow \{-1,+1\}$
  %% corresponding to $f$ but with outputs labeled $\pm 1$; i.e.\ $f'=1-2f$.   Decision
  %% trees for $f'$ and for $f$ are in bijection with a trivial
  %% relabelling; in particular, rank, depth, size are the same for both
  %% functions.
  Recall that $\sparn(f)$ refers to the number of non-zero Fourier coefficients in the expansion of $f$ apart from $\hat{f}(\emptyset)$. We prove the theorem by induction on $\sparn(f)$.
  
    When $\sparn(f)<1$, $f$  is  a constant function with $\Depth(f)=\Rank(f)=0$, and the inequality holds.

Now assume that $\sparn(f) \ge 1$.  We give a recursive construction of a decision tree for $f$.

  Choose a subcube $J=(S,\rho)$ of minimum co-dimension $|S|=K(f)$ on which $f$  becomes constant. By \cref{lemma:kill}, $|S| \le \Rank(f)$. Start by querying all the variables indexed in $S$ in any order. When the outcome of all these queries matches $\rho$, the function becomes a constant and the tree terminates at this leaf. On any other outcome $\rho'$, the function is restricted to the subcube $J'=(S,\rho')$, and by  \cref{thm:rest-sparsity}, $\sparn(f|J') \le \sparn(f)/2$. Proceed recursively  to build the decision tree of $f|J'$ which is then attached to this leaf.

  Each stage in the recursion makes at most $\Rank(f)$ queries and halves sparsity $\sparn$. After at most $1+\log  \sparn(f)$ stages, the sparsity of the restricted function drops to below 1 and the function becomes a constant. 

  Thus the overall depth of the entire tree is bounded by $\Rank(f)\cdot(1+ \log \sparn(f))$, which is at most
$\Rank(f)\cdot(1+ \log \spar(f))$,  as claimed.
%  
%We give a recursive construction of the decision tree. Since $\skill(f)\le \Rank(f)$, choose a subcube $J=(S,\rho)$ of co-dimension less than equal to $\Rank(f)$ on which $f$ becomes constant. Such a subcube exists because of \cref{lemma:kill}. Start by querying all the variables indexed in $S$ in any order. The function becomes constant on at least one of the outcomes of these queries, where variables are answered according to $\rho$. On the other outcome, the function reduces to $f$ restricted to subcube $J'=(S,\rho')$ where the answerd queries determine $\rho'$. Using \cref{thm:rest-sparsity}, the sparsity of $f|J'$ is at most $\spar(f)/2$. Next, recursively build the decision tree for $f|J'$ by repeating the process. Since after each batch of queries, sparsity reduces by a factor of $2$, sparsity becomes less than $1$, and function becomes constant, in $\log (spar(f)) + 1$ many such iterations. In each iteration we query $\Rank(f)$ many variables, this follows as $\max_{\text{subcube } J } \skill(f|J) \le \Rank(f)$ from \cref{lemma:kill}, giving us a decision tree of claimed depth $\Rank(f)(\log (\spar(f)) +1 )$.
\end{proof}
%\end{shaded}

\subsection{Relation between Rank and Certificate Complexity}
\label{subsec:rank-cert}
The certificate complexity and decision tree complexity are known to be related as follows.
\begin{proposition}[\cite{blum1987generic},\cite{hartmanis1991one},\cite{tardos1989query}, see also Theorem 14.3 in \cite{Jukna-BFCbook}]\label{prop:depth-cert}
For every boolean function $f: \boolfn{n}$,
$$\Cert(f) \le \Depth(f)  \le \CertZ(f)\CertO(f)$$
\end{proposition}
Both these inequalities are tight; the first for the
$\OR$ and $\AND$ functions, and the second for the $\Tribes_{n,m}$ and $\dTribes_{n,m}$ functions.
(For $\dTribes_{n,m}$,  $\CertZ(\dTribes_{n,m}) = m$, $\CertO(\dTribes_{n,m}) = n$ and $\Depth(\dTribes_{n,m})=nm$, see e.g.\ Exercise 14.1 in \cite{Jukna-BFCbook}.)

Since $\Rank \le \Depth$, the same upper bound also holds for $\Rank$ as well.
But it is far from tight for the $\Tribes_{n,m}$ function.  In fact, the upper bound can be improved in general. Adapting the construction given in the proof of \cref{prop:depth-cert} slightly, we  show the following.
\begin{lemma}\label{lem:rank-cert}
  For every Boolean function $f: \boolfn{n}$,
  $$ \Rank(f) \le (\CertZ(f)-1)(\CertO(f)-1) + 1$$
  Moreover, the inequality is tight as witnessed by $\AND$ and $\OR$ functions.
\end{lemma}
\begin{proof}
The inequality holds trivially for constant functions since for such functions, $\Rank=\CertZ=\CertO=0$. So assume $f$   is not constant. The proof is by induction on $\CertO(f)$. 

Base Case: $\CertO(f)=1$. Let $S\subseteq [n]$ be the set of indices  that are 1-certificates for some $a\in f^{-1}(1)$. We construct a  decision tree by querying all the variables indexed in $S$. For each such query, one outcome immediately leads to a 1-leaf (by definition of certificate), and we continue along the other outcome. If all variables indexed in $S$ are queried without reaching a 1-leaf, the restricted function is 0 everywhere and so we create a 0-leaf.  This gives a rank-1 decision tree computing $f$. 

For the inductive step, assume $\Rank(g) \le (\CertZ(g)-1)(\CertO(g)-1) + 1$ is true for all $g$ with $\CertO(g)\le k$. Let $f$ satisfy $\CertO(f)=k+1$. Pick an $a\in f^{-1}(0)$  and a mimimum-size 0-certificate $S$ for $a$. Without loss of generality, assume that $S=\{x_1,x_2,\ldots,x_\ell\}$ for some $\ell=|S|\le \CertZ(f)$.  Now, take a complete decision tree $T_0$ of depth $l$ on these $l$ variables. Each of its leaves corresponds to the unique input $c=(c_1,c_2,...,c_l) \in \bool^l $ reaching this leaf. At each such leaf, attach a minimal rank decision tree $T_c$ for the subfunction $f_c\triangleq f(c_1,c_2,...,c_l,x_{l+1},...,x_n)$. This gives a decision tree $T$ for $f$. We now analyse its rank. 

For at least one input $c$, we know that $f_c$ is the constant function $0$. For all leaves where $f_c$ is not 0, $\CertZ(f_c) \le \CertZ(f)$ since certificate size cannot increase by assigning some variables. Further, $\CertO(f_c) \le \CertO(f)-1$; this because of the well-known fact (see e.g.\ \cite{Jukna-BFCbook}) that every pair of a 0-certificate and a 1-certificate for $f$ have at least one common variable, and $T_0$ has queried all variables from a 0-certificate. Hence, by induction, for each $c$ with $f_c\neq 0$,
$\Rank(T_c) \leq (\CertZ(f_c)-1)(k-1) +1\leq (\CertZ(f)-1)(k-1) +1$. 
Thus $T$ is obtained from a rank-$\ell$ tree $T_0$ (with $\ell \le \CertZ(f)$) by attaching a tree of rank 0 to at least one leaf, and attaching trees of rank at most $(\CertZ(f)-1)(k-1) +1$ to all leaves.  From \cref{prop:compose_rank_dt}, we conclude that
$\Rank(f)\leq \Rank(T)\le ((\CertZ(f)-1)(k-1) +1)+(l-1) \leq (\CertZ(f)-1)(\CertO(f)-1) + 1$.   
\end{proof}

From \cref{thm:rank-tribes}, we see that the lower bound on $\Depth$ in \cref{prop:depth-cert} does not hold for $\Rank$; for $m>n$, $\Rank(\dTribes_{n,m})=n < m = C(\dTribes_{n,m})$. However, $\min\{\CertZ(\dTribes_{n,m}),\CertO(\dTribes_{n,m})\} = n = \Rank(\dTribes_{n,m})$. Further, for all the functions listed in \cref{tab:tabulation}, $\Rank(f)$ is at least as large as $\min\{\CertZ(f),\CertO(f)\}$.
However, even this is  not a lower bound in general.

\begin{restatable}{lemma}{lemMinCNotLB}\label{lem:minC-not-a-lb-for-rank}
  $\min\{\CertZ(f),\CertO(f)\}$ is not a lower bound on $\Rank(f)$; for the symmetric function $f=\Maj_n \vee \Parity_n$, when $n>4$,  $\Rank(f) < \min\{\CertZ(f),\CertO(f)\}$.
\end{restatable}
\begin{proof}
  Let $f$ be the function $\Maj_n \vee \Parity_n$, for $n>4$. Then  $f(0^n)=0$ and $\CertZ(f,0^n)=n$, and $f(10^{n-1})=1$ and $\CertO(f,10^{n-1})=n$. Also, $f$ is symmetric, with $\Gap(f)=n/2$, so by \cref{lem:ABDORU}, $\Rank(f)=n/2$. 
\end{proof}

The average certificate complexity is also not directly related to rank.
\begin{lemma}\label{lem:avgC-not-a-lb-for-rank}
Average certificate complexity is neither a upper bound nor a lower bound on the rank of a function; there exist functions $f$ and $g$, such that $\Rank(f) < \aCert(f)$ and $\aCert(g) < \Rank(g)$.
\end{lemma}
\begin{proof}
  Let $f$ be the $\AND_n$ function for $n\ge 2$; we know that $\Rank(f)=1$. Since the $1$-certificate has length $n$ and all minimal $0$-certificates have length $0$, the average certificate complexity of $f$ is $\aCert(f)=2^{-n}.n + (1-2^{-n}).1= 1+ 2^{-n}(n-1)$. %This gives us $\aCert(f)>\Rank(f)$ as rank of the $\AND_n$ function is $1$.\\
  
  Consider $g=\dTribes_{n,2}$ for $n>2$. By \cref{thm:rank-tribes}, $\Rank(g)=n$. Since $|g^{-1}(1)|=3^n$  and each minimal 1-certificate has length $n$, and since $|g^{-1}(0)|=4^n-3^n$  and each minimal 0-certificate has length $2$, we see that \[\aCert(g) = \left(\frac{3}{4}\right)^n \cdot n + \left[1- \left(\frac{3}{4}\right)^n\right] \cdot 2 < n =\Rank(g).\]
  
%  All the $1$-certificates of $g$ have length $n$, and all the $0$-certificates have length $2$. Hence the average certificate complexity of $g$ is $\aCert(g)=(\frac{3}{4})^n .n + (1- \frac{3}{4})^n. 2$. Since rank of $g$ is $n$, see \cref{thm:rank-tribes}, we have $\aCert(g)< \Rank(g)$.

  For a larger gap between $\Rank$ and $\aCert$, consider the function $h=\AND_n\circ \Parity_n$. From \cref{thm:rank-tribes}, $\Rank(h) = n(n-1)+1$.  There are $2^{(n-1)n}$ 1-inputs, and all the $1$-certificates have length $n^2$. Also, all minimal $0$-certificates have length $n$. Hence $\aCert(h)= 2^{-n}n^2 + (1-2^{-n})n = n + o(1)$. 
\end{proof}

What can be shown in terms of certificate complexity and rank is the following:
\begin{lemma}\label{lem:mincert-rank}
  For every Boolean function $f$,
  $\minCert(f) \le \Rank(f)$. This is tight for $\OR_n$. 
\end{lemma}
\begin{proof}
  Let $T$ be a rank-optimal decision tree for $f$. Since the variables queried in any root-to-leaf path in $T$ form a $0$ or $1$-certificate for $f$, we know that depth of each leaf in $T$ must be at least $\minCert(f)$.
  By \cref{prop:prop_rank_tree}(\ref{item-leaf-depth-rank}),  $\Rank(f) = \Rank(T)\ge \minCert(f)$. 
\end{proof}

\cref{lem:rank-cert} and \cref{lem:mincert-rank} give these bounds sandwiching
$\Rank(f)$:
\begin{theorem}\label{thm:rank-cert-bounds}
  $\minCert(f) \le \Rank(f)  \le (\CertZ(f)-1)(\CertO(f)-1) + 1 \le (\Cert(f)-1)^2+1$.
\end{theorem}

As mentioned in \cref{lem:ABDORU}, for symmetric functions the rank is completely characterised in terms of $\Gap$ of $f$. How does $\Gap$ relate to certificate complexity for such functions? It turns out that certificate complexity is characterized not by $\Gap$ but by $\minGap$. Using this relation, the upper bound on $\Rank(f)$ from \cref{lem:rank-cert} can be improved for symmetric functions to $\Cert(f)$.
\begin{lemma}\label{lem:symm-cert}
  For every symmetric Boolean function $f$ on $n$ variables, $\Cert(f)=n-\minGap(f)$ and $n-\Cert(f)+1 \le \Rank(f)
  \le \Cert(f)$. Both the inequalities on rank are tight for $\Maj_{2k+1}$.  
\end{lemma}
\begin{proof}
 We first show $\Cert(f)=n-\minGap(f)$. Consider any interval $[a,b]$ such that $f_{a-1}\neq f_a=f_{a+1}=...=f_b\neq f_{b+1}$. Let $x$ be any input with Hamming weight in the interval $[a,b]$. We show that $C(f,x)=n-(b-a)$.
\begin{enumerate}
\item   Pick any $S\subseteq[n]$ containing exactly $a$  bit positions where $x$ is 1, and exactly $n-b$ bit positions where $x$ is 0. Any $y$ agreeing with $x$ on $S$ has Hamming weight in $[a,b]$, and hence $f(y)=f(x)$. Thus $S$ is a certificate for $x$. Hence $C(f,x)\le n-(b-a)$.
\item   Let $S\subseteq [n]$ be any certificate for $x$. Suppose $S$ contains fewer than $a$ bit positions where $x$ is 1. Then there is an  input $y$ that agrees with $x$ on $S$ and has Hamming weight exactly $a-1$. (Flip some of the 1s from $x$ that are not indexed in $S$.) So $f(y) \neq f(x)$, contradicting the fact that $S$ is a certificate for $x$. Similarly, if $S$ contains fewer that $n-b$ bit positions where $x$ is 0, then there is an input $z$ that agrees with $x$ on $S$ and has Hamming weight exactly $b+1$. So $f(z) \neq f(x)$, contradicting the fact that $S$ is a certificate for $x$. 

  Thus any certificate for $x$ must have at least $a+(n-b)$ positions; hence  $C(f,x) \ge n-(b-a)$. 
\end{enumerate}
Since the argument above works for any interval $[a,b]$ where $f$ is constant, we conclude that $\Cert(f) = n - \minGap(f)$.

Next, observe that $\Gap(f) + \minGap(f) \leq n-1$. Hence, 
  \[n-\Cert(f)+1 = \minGap(f)+1 \le n-Gap(f)=\Rank(f) \le n-\minGap(f) = \Cert(f). \]

  As seen from \cref{tab:tabulation}, these bounds on $\Rank$ are tight for $\Maj_{2k+1}$. 
\end{proof}

Even for the (non-symmetric) functions in \cref{thm:rank-tribes}, $\Rank(f) \le \Cert(f)$. However, this is not true in general.

\iffalse
%\begin{lemma}\label{lem:cert-not-ub}
%  Certificate Complexity does not always bound $\Rank$ from above; for the
 % function $f = \Maj_{2k+1}\circ \Maj_{2k+1}$, $\Cert(f) <  \Rank(f)$.
%\end{lemma}
The proof is deferred to \cref{sec:rank-composed}, where we develop techniques to bound the rank of composed functions. We also give, in \cref{sec:game-proofs}, a proof based on the Prover-Delayer game characterisation from \cref{thm:game-rank}.
\fi

%\begin{shaded}
\begin{lemma}\label{lem:cert-not-ub}
 Certificate Complexity does not always bound $\Rank$ from above; for $k\ge1$ and $n=4^k$ the
 function $f = (\AND_2\circ \OR_2)^{\otimes k}$ on $n$ variables has $\Rank(f) = \Omega(Cert(f)^2)$.
\end{lemma}
This lemma shows that the relation between rank and certificate complexity (from \cref{lem:rank-cert}) is optimal upto constant factors.
The proof of the  lemma is deferred to the end of \cref{subsec:size-lb}, before which we develop techniques to bound the rank of composed functions.
%\end{shaded}

\section{Rank of Composed and Iterated Composed functions}
\label{sec:rank-composed}
In this section we study the rank for composed functions. For composed functions, $f\circ g$, decision tree complexity $\Depth$ is known to behave very nicely.
\begin{proposition}[\cite{Montanaro-cj14}]\label{prop:depth-compose}
    For Boolean functions $f,g$, 
  $\Depth(f\circ g)=\Depth(f)\Depth(g)$.
\end{proposition}
We want to explore how far something similar can be deduced about $\Rank(f\circ g)$.
The first thing to note is that a direct analogue in terms of $\Rank$ alone is ruled out.
\begin{lemma}\label{lem:compose-rank-example}
For general Boolean functions $f$ and $g$, $\Rank(f\circ g)$ cannot be bounded by any function of $\Rank(f)$
and $\Rank(g)$ alone. 
\end{lemma}
\begin{proof}
Let $f=\AND_n$ and $g=\OR_n$. Then $\Rank(f)=\Rank(g)=1$. But $\Rank(f\circ g) = \Rank(\dTribes_n) = n$, as seen in \cref{thm:rank-tribes}. 
\end{proof}

For $f\circ g$, let $T_f$, $T_g$ be decision trees for $f$,
$g$ respectively. One way to construct a decision tree for $f\circ g$
is to start with $T_f$, inflate each internal node $u$ of $T_f$ into a
copy of $T_g$ on the appropriate inputs, and attach the left and the right subtree of $u$ 
as appropriate at the leaves of this copy of
$T_g$. By \cref{prop:depth-compose}, the decision tree thus obtained for $f\circ g$ is optimal for $\Depth$ if one start with depth-optimal trees $T_f$ and $T_g$ for $f$ and $g$ respectively. In terms of rank, we  can also show that the rank of the  decision tree so constructed is bounded above by  $\Depth(T_f)\Rank(T_g)=\wDepth(f,[r,r,\ldots , r])$, where $r=\Rank(T_g)$. (This is  the construction used in the proofs of \cref{lem:rank-tribes-ub,lem:and-parity-ub},  where further properties of the $\Parity$ function are used  to show that the resulting tree's rank is even smaller than
$\Depth(f)\Rank(g)$.) In fact, we show below (\cref{thm:compose-rank-ub}) that this holds more generally, when different functions are used in the composition. While this is a relatively straightforward generalisation here, it is necessary to consider such compositions for the lower bound we establish further on in this section.
\begin{restatable}{theorem}{thmComposeRankUb}\label{thm:compose-rank-ub}
  For non-constant boolean functions $g_1, \ldots , g_n$ with $\Rank(g_i)=r_i$, and for $n$-variate  non-constant booolean function $f$,
$$\Rank(f\circ(g_1,g_2,...,g_n)) \le \wDepth(f,[r_1,r_2,...,r_n] ).$$
\end{restatable}
\begin{proof}
  Let $h$ denote the function $f\circ(g_1,g_2,...,g_n)$. For $i\in[n]$, let $m_i$ be the arity of $g_i$.  We call $x_{i,1}, x_{i,2}, \ldots , x_{i,m_i}$ the $i$th block of variables of $h$; $g_i$ is evaluated on this block.   
  Let $T_f$ be any decision tree for $f$. % with weights $[r_1,r_2,...,r_n]$. 
%The recursive construction analysis should not use depth-optimality of $T_f$ since not both subtrees may be optimal. At the end we can pick the optimal tree.
  For each $i\in [n]$, let $T_{g_i}$ be a rank-optimal tree for $g_i$.  Consider the following recursive construction of a decision tree $T_{h}$ for $h$. 
\begin{enumerate}
\item Base Case: $\Depth(T_f)=0$. Then $f$ and $h$ are the same constant function, so set $T_h=T_f$. 
\item Recursion Step: $\Depth(T_f)\ge 1$. Let $x_{i}$ be the variable queried at the root node of $T_f$, and let $T_0$ and $T_1$ be the left and the right subtree of $T_f$, computing functions $f_0,f_1$ respectively. For notational convenience, we still view $f_0,f_1$ as functions on $n$ variables, although they do not depend on their $i$th variable. 
  Recursively construct, for $b\in\bool$, the trees $T'_b$ computing $f_b\circ(g_1,\ldots,g_{i-1},b,g_{i+1},\ldots,g_n)$ on the variables $x_{k,\ell}$ for $k\neq i$.   Starting with the tree $T_{g_i}$ on the $i$th block of variables, attach tree $T'_b$ to each leaf labeled $b$ to obtain the tree $T_h$. 
\end{enumerate}
From the construction, it is obvious that $T_{h}$ is a decision tree for $f\circ (g_1,\ldots ,g_n)$. It remains to analyse the rank of $T_{h}$. Proceeding by
induction on $\Depth(T_f)$, we show that $\Rank(T_{h}) \le D_w(T_f,[r_1,r_2,...,r_n] )$. 
\begin{enumerate}
\item Base Case: $\Depth(T_f)=0$. Then %$f=h$ is a constant function,  so
$T_h=T_f$, so   $\Rank(T_{h}) = D_w(T_f,[r_1,r_2,...,r_n] )=0$.
\item Induction: $\Depth(T_f) \ge 1$.
%  Let $x_{i}$ be the label of the root node of $T_f$ and $T_0$ and $T_1$ be the left and the right subtree. Using  \cref{prop:compose_rank_dt}, we have that $\Rank(T_{h}) \le \Rank(T_{g_i}) + \max(\Rank(T_0'),\Rank(T_1'))$ where $T_0'$ and $T_1'$ are as defined in the construction. Using induction hypothesis, we have that $\Rank(T_0') \le D_w(T_0,[r_1,r_2,...,r_n] ) $ and $\Rank(T_1') \le D_w(T_1,[r_1,r_2,...,r_n] )$. Thus
\begin{align*}
  \Rank(T_{h}) &\le \Rank(T_{g_i}) + \max\{\Rank(T_0'),\Rank(T_1')\}
  \quad (\textrm{by \cref{prop:compose_rank_dt}})\\
  &=  r_i + \max_{b\in\bool}\{\Rank(T_b')\} \\
  &\le r_i + \max_{b\in\bool}\{D_w(T_b,[r_1,r_2,...,r_n])\}
    \quad (\textrm{by induction})\\
  &= D_w(T_f,[r_1,r_2,...,r_n] )
    \quad \text{by definition of $D_w$} 
%  \quad \text{by definition of $D_w$  and $T_f$ being optimal} 
\end{align*}
%which proves the induction case.
\end{enumerate}
Picking $T_f$ to be a tree for $f$ that is optimal with respect to weights $[r_1,r_2,...,r_n]$ , we obtain $\Rank(h) \le \Rank(T_{h}) \le D_w(T_f,[r_1,r_2,...,r_n] ) = D_w(f,[r_1,r_2,...,r_n] )$.
\end{proof}

The really interesting question, however, is whether we can show a good lower bound for the rank of a composed function. This will help us understand how  good is the upper bound in \cref{thm:compose-rank-ub}. 
To begin with, note that for non-constant Boolean functions $f,g$, both $f$ and $g$ are sub-functions of $f\circ g$. Hence \cref{prop:rank_subfn} implies the following.
\begin{proposition}\label{prop:compose-rank-max-lb}
  For  non-constant boolean functions  $f,g$,  
$$\Rank(f\circ g) \ge \max\{\Rank(f),\Rank(g)\}.$$
\end{proposition}

A better lower bound in terms of weighted depth complexity of $f$ is given below. This generalises the lower bounds from \cref{lem:rank-tribes-lb,lem:and-parity-lb}. The proofs of those lemmas crucially used nice symmetry properties of the inner function, whereas the bound below applies for any non-constant inner function. It is significantly weaker than the bound from \cref{lem:rank-tribes-lb} but matches that from \cref{lem:and-parity-lb}.
\begin{restatable}{theorem}{thmComposeRankLb}\label{thm:compose-rank-lb}
  For non-constant boolean functions $g_1, \ldots , g_n$ with $\Rank(g_i)=r_i$, and for $n$-variate  non-constant boolean function $f$,
%$$\Rank(f\circ(g_1,g_2,...,g_n)) \ge \wDepth(f,[r_1-1,r_2-1,...,r_n-1] ) + 1 \ge \wDepth(f,[r_1,r_2,...,r_n] ) - (n-1).$$
  \begin{align*}
    \Rank(f\circ(g_1,g_2,...,g_n)) & \ge \wDepth(f,[r_1-1,r_2-1,...,r_n-1] ) + 1 \\
    & \ge \wDepth(f,[r_1,r_2,...,r_n] ) - (n-1).
  \end{align*}
\end{restatable}
\begin{proof}
  The second inequality above is straightforward: let $T$ be a decision tree for $f$ that is optimal with respect to weights $r_1-1,\ldots ,r_n-1$. Since $T$ can be assumed to be reduced, repeated application of \cref{fact:wtd-dec-tree} shows that the  depth of $T$ with respect to weights $r_1,\ldots ,r_n$ increases by at most $n$. Thus $\wDepth(f,[r_1,\ldots,r_n]) \le \wDepth(T,[r_1,\ldots,r_n])  \le \wDepth(T,[r_1-1,\ldots,r_n-1])+n=\wDepth(f,[r_1-1,\ldots,r_n-1])+n$, giving the claimed inequality. 
  
  We now turn our attention to the first inequality, which is not so straightforward.   We prove it by induction on $n$. 
  Let $h$ denote the function $f\circ(g_1,g_2,...,g_n)$. For $i\in[n]$, let $m_i$ be the arity of $g_i$.  We call $x_{i,1}, x_{i,2}, \ldots , x_{i,m_i}$ the $i$th block of variables of $h$; $g_i$ is evaluated on this block.

%\begin{enumerate}
  %\item
  In the base case, $n=1$.  Since $f$ is non-constant, $f$ can either be $x$ or $\neg x$; accordingly, $h$ is either $g_1$ or  $\neg g_1$.  So $D_w(f,[r_1-1])=r_1-1$ and $\Rank(h)=\Rank(g_1)=r_1$, and the inequality holds.

 %\item
 For the inductive step, when $n>1$, we proceed by induction on $M=\sum_{i=1}^n m_i$. 
%\begin{enumerate}
 %\item
 In the base case, $M=n$, and each $m_i$ is equal to $1$. Since all $g_i$'s  are non-constant, $r_i=1$ for all $i$.  So  $D_w(f,[r_1-1,r_2-1,...,r_n-1])+1 = D_w(f,[0,0,...,0])+1=1$. 
Since all $r_i$'s are $1$, each $g_i$'s is either $x_{i,1}$ or $\neg x_{i,1}$, 
Thus $h$ is the same as $f$ upto renaming of the literals. Hence
$\Rank(h)=\Rank(f)\ge 1$.

%\item
For the inductive step, $M>n>1$. 
Take a rank-optimal decision tree $T_h$ for $h$. We want to show that
$\wDepth(f,[r_1-1,\ldots, r_n-1]) \le \Rank(T_h)-1$. 
Without loss of generality, let  $x_{1,1}$ be the variable queried at the root. Let  $T_0$ and $T_1$ be the left and the right subtree of $T_h$. For $b\in\bool$,
let $g_1^b$ be the subfunction of $g_1$ when $x_{1,1}$ is set to $b$. Note that $T_b$  computes $h_b\triangleq f\circ(g_1^b,g_2,...,g_n)$, a function on $M-1$ variables. We would like to use induction to deduce information about $\Rank(T_b)$. However, $g_1^b$ may be a constant function, and then induction does not apply. So we do a case analysis on whether or not $g_1^0$ and $g_1^1$ are constant  functions; this case analysis is lengthy and tedious  but most cases are straightforward.
%For the inductive step, $M>n>1$. This is a lengthy and tedious case analysis. Take a rank optimal decision tree, $T_h$, for $h$. Let w.l.o.g $x_{11}$ be the variable queried at the root  and $T_0$ and $T_1$ be the left and the right subtree of $T_h$. Let $g_1^0$ and $g_1^1$ be $g_1$ with $x_{11}$ set to $0$ and $1$ respectively. Note that $T_0$  computes $h_0\triangleq f\circ(g_1^0,g_2,...,g_n) $  and $T_1$ computes $h_1\triangleq f\circ(g_1^1,g_2,...,g_n)$. We do case analysis on whether $g_1^0$ and $g_1^1$ is constant or non-constant functions.
\begin{itemize}
\item Case 1: Both $g_1^0$ and $g_1^1$ are constant functions. Since $g_1$ is non-constant, $g_1^0 \neq g_1^1$, and $r_1=\Rank(g_1)=1$. Assume that $g_1^0=0$ and $g_1^1=1$; the argument for the other case is identical.
  For $b\in\bool$, let $f_b$ be the function $f(b,x_2,\ldots,x_n)$; then
  $h_b=f_b\circ(g_2,\ldots ,g_n)$. View $f_b$ as functions on $n-1$ variables.
\begin{itemize}
\item Case 1a: Both $f_0$ and $f_1$ are constant functions. Then  $f$ is either $x_1$ or $\neg x_1$, so $\wDepth(f,[r_1-1,r_2-1,...,r_n-1]) =
  \wDepth(f,[0,r_2-1,...,r_n-1]) = 0$. Also, in this case, $h$ is either $x_{1,1}$ or $\neg x_{1,1}$, so $\Rank(h)=1$. Hence the inequality holds. 
\item Case 1b:  Exactly one of $f_0$ and $f_1$ is a constant function; without loss of generality, let $f_0$ be a constant function. First, observe that for any weights $w_2,\ldots,w_n$, $D_w(f,[0,w_2,...,w_n]) \le D_w(f_1,[w_2,...,w_n])$: we  can obtain a decision tree for $f$ witnessing this by first querying $x_1$, making the $x_1=0$ child a leaf labeled $f_0$, and attaching the optimal tree for $f_1$ on the $x_1=1$ branch. Second, note that since $f_1$ and all $g_i$ are non-constant, so is $h_1$. Now
\begin{align*}
  \Rank(h)&=\Rank(h_1) && \text{since $\Rank(h_0)=0$} \\
  &\ge D_w(f_1,[r_2-1,...,r_n-1])+1 && \text{by induction hypothesis on $n$}\\
%&= D_w(f_1,[0,r_2-1,...,r_n-1])+1\\
&\ge D_w(f,[0,r_2-1,...,r_n-1])+1 && \text{by first  observation above}\\
&= D_w(f,[r_1-1,r_2-1,...,r_n-1])+1  && \text{since $r_1=1$}
\end{align*}
\item Case 1c:  Both $f_0$ and $f_1$ are non-constant functions. 
\begin{align*}
\Rank(h) &\ge \max(\Rank(h_0),\Rank(h_1))\\
&\ge \max_{b\in\bool}\{D_w(f_b,[r_2-1,...,r_n-1])\}  + 1 && \text{by induction hypothesis on $n$}\\
&\ge D_w(f,[0,r_2-1,...,r_n-1]) + 1 && \text{by def.\ of weighted depth}\\
&&& \text{of a tree querying $x_1$ first} \\
&= D_w(f,[r_1-1,r_2-1,...,r_n-1]) + 1 && \text{since $r_1=1$}
\end{align*}
\end{itemize}

\item Case 2: One of $g_1^0$ and $g_1^1$ is a constant function; assume without loss of generality that $g_1^0$ be constant. In this case, we can conclude  that $\Rank(g_1) = \Rank(g_1^1)$:  $\Rank(g_1^1) \le \Rank(g_1)$ by \cref{prop:rank_subfn}, and  $\Rank(g_1) \le \Rank(g_1^1)$ as witnessed by a decision tree for $g_1$ that queries $x_{1,1}$ first, sets the $x_{1,1}=0$ branch to a leaf labeled $g_1^0$, and attaches an optimal tree for $g_1^1$ on the other branch. 
Now
\begin{align*}
\Rank(h)&\ge \Rank(h_1)\\
&\ge D_w(f,[\Rank(g_1^1)-1,r_2-1,...,r_n-1])+1 && \text{by induction on $M$}\\
&= D_w(f,[r_1-1,r_2-1,...,r_n-1])+1 && \text{since $\Rank(g_1^1)=\Rank(g_1)$}
\end{align*}

\item Case 3:  Both $g_1^0$ and $g_1^1$ are non-constant functions. Let  $r_1^b=\Rank(g_1^b)\ge 1$.  A decision tree for $g_1$ that queries $x_{1,1}$ first and then uses optimal trees for $g_1^0$ and $g_1^1$ has rank $R \ge r_1$ and witnesses that $1+\max\{r_1^0,r_1^1\} \ge R \ge r_1$. (Note that $R$ may be more than $r_1$, since a rank-optimal tree for $g_1$ may not query $x_{1,1}$ first.)  

\begin{itemize}
\item Case 3a: $\max_b\{r_1^b\} = r_1-1$. Then $R=1+\max\{r_1^0,r_1^1\}$, which can only happen if  $r_1^0=r_1^1$, and hence $r_1^0=r_1^1=r_1-1$.  We can further conclude that $r_1 \ge 2$. Indeed, if $r_1=1$, then  $r_1-1=r_1^0=r_1^1=0$,
  contradicting the fact that we are in Case 3. %making $g_1^0$ and $g_1^1$ constant functions and taking us back to case Case 1, not Case 3.

For $b\in\bool$,
\begin{align*}
  \Rank(h_b) &= \Rank(f\circ(g_1^b,g_2,\ldots ,g_n))\\
  &\ge \wDepth(f,[r_1^b-1, r_2-1, \ldots, r_n-1])+1 \quad \text{by induction on $M$} \\
  &= \wDepth(f,[r_1-2,r_2-1,\ldots,r_n-1])+1 \quad \text{since $r_1-1= r_1^b$}.\\ 
\text{Hence~}  \Rank(h) &\ge 1 + \min_b \Rank(h_b) \\
  &\ge \wDepth(f,[r_1-2,r_2-1,\ldots,r_n-1])+2 \quad \text{derivation above} \\
  &\ge \wDepth(f,[r_1-1, r_2-1, \ldots, r_n-1])+1 \quad \text{by \cref{fact:wtd-dec-tree}}
\end{align*}

\item Case 3b: $\max_b\{r_1^b\} > r_1-1$. So $\max_b\{r_1^b\} \ge r_1$. 
\begin{align*}
  \Rank(h) &\ge \max_b \Rank(h_b) \\
  &\ge \max_b \wDepth(f,[r_1^b-1,r_2-1,\ldots,r_n-1])+1 \quad \text{by induction on $M$} \\
  &\ge \wDepth(f,[r_1-1,r_2-1,\ldots,r_n-1])+1 \quad \text{since $\max_b\{r_1^b\} \ge r_1$} 
\end{align*}
\end{itemize} 
\end{itemize}

%\end{enumerate} 
%\end{enumerate}
This completes the inductive step for $M>n>1$ and completes the entire proof.
\end{proof}

From \cref{thm:rank-tribes,thm:compose-rank-ub,thm:compose-rank-lb}, we obtain the following:
\begin{theorem}\label{thm:compose-rank-bounds}
  For non-constant boolean functions $f,g$,
\[ \Depth(f) (\Rank(g)-1) +1 \le \Rank(f\circ g) \le \Depth(f) \Rank(g). \]
Both inequalities are tight; the first for $\AND_n \circ \Parity_m$ 
and the second for $\Tribes_n$ and  $\dTribes_n$.
\end{theorem}
It is worth noting that in the above bounds,  the role of $\Rank$ and $\Depth$ cannot be exchanged. With $f=\AND_n$ and $g=\Parity_n$, $\Rank(f)\Depth(g)=n < n(n-1)+1 \le \Rank(f\circ g)$, and $\Rank(g\circ f) \le n < n(n-1)+1 = \Rank(g)(\Depth(f)-1)+1 = \Depth(f)(\Rank(g)-1)+1$. 

Since any non-constant symmetric  function is evasive (\cref{prop:symm_evasive}), from \cref{thm:compose-rank-ub,thm:compose-rank-lb}, we obtain the following:
\begin{corollary}
  For non-constant boolean functions $g_1, \ldots , g_n$ with $\Rank(g_i)=r_i$, and for $n$-variate symmetric non-constant booolean function $f$,
$$\sum_i r_i - (n-1) \le \Rank(f\circ(g_1,g_2,...,g_n))\le \sum_i r_i .$$
\end{corollary}

\iffalse
Using \cref{thm:compose-rank-bounds}, we can now complete  the proof of \cref{lem:cert-not-ub}.
\begin{proof}(of \cref{lem:cert-not-ub})
Consider the composed function $f=\Maj_{2k+1}\circ\Maj_{2k+1}$.  Note that from the lower bound in \cref{thm:compose-rank-bounds}, and the entries in \cref{tab:tabulation}, $\Rank(\Maj_{2k+1}\circ\Maj_{2k+1}) \ge (2k+1)k+1$. On the other hand, it is straightforward to verify that $\Cert(f)=(k+1)^2$. Thus for $k> 1$, $\Rank(f) > \Cert(f)$.
\end{proof}
\fi
%\begin{shaded}
For iterated composed functions, we obtain the following corollary.
\begin{corollary}\label{corr:iterated-rank}
 For $k\ge 1$ and non-constant boolean functions $f$,
\[ \Depth(f)^{k-1} (\Rank(f)-1) +1 \le \Rank(f^{\otimes k}) \le \Depth(f)^{k-1} \Rank(f). \]
\end{corollary}
\begin{proof} 
The result follows from \cref{thm:compose-rank-bounds}  applying induction on $k$. The base case, k=1, is straightforward. For the induction step, $k>1$, applying the recursive definition of iterated composed functions, we have
\begin{align*}
 \Rank(f^{\otimes k}) &= \Rank(f\circ f^{\otimes (k-1)})\\
 &\ge \Depth(f)(\Rank(f^{\otimes (k-1)}) -1)+1 \quad \text{by \cref{thm:compose-rank-bounds}} \\
 &\ge \Depth(f)(\Depth(f)^{k-2}(\Rank(f)-1))+1 \quad \text{by induction on $k$}\\
 &= \Depth(f)^{k-1}(\Rank(f)-1)+1.
 \\
%\end{align*}
%And  
%\begin{align*}
 \Rank(f^{\otimes k}) &= \Rank(f\circ f^{\otimes (k-1)})\\
 &\le \Depth(f)\Rank(f^{\otimes (k-1)}) \quad \text{by \cref{thm:compose-rank-bounds}} \\
 &\le \Depth(f)(\Depth(f)^{k-2}(\Rank(f)) \quad \text{by induction on $k$}\\
 &= \Depth(f)^{k-1}\Rank(f).
\end{align*}
\end{proof}
%\end{shaded}

\section{Applications}\label{sec:application}
In this section, we give some applications of our results and methods. We first show how to obtain tight lower bounds on $\log \Size$ for composed functions using the rank lower bound  from \cref{thm:compose-rank-bounds}. Next, we relate rank to query complexity in more general decision trees, namely $\CONJ$ decision trees, and show that rank (for ordinary decision trees) characterizes query complexity in this model up to $\log n$ factors.

\subsection{Tight lower bounds for $\log\Size$ for Composed functions}
\label{subsec:size-lb}
It was shown in  \cite{EH-IC1989} that  every boolean function $f$ in $n$ variables has a decision tree of size at most $\exp(O(\log n \log^2 N(f))$, where $N(f)$ is the total number of monomials in the minimal $\dnf$ for $f$ and $\neg f$. Later,  in  \cite{jukna1999p}, this relation was proved to be optimal up to $\log n$ factor. To prove this, the authors of \cite{jukna1999p} showed that iterated $\AND_2\circ \OR_2$ and iterated $\Maj_3$ on $n=4^k$ and $n=3^k$ variables  require decision trees of size $\exp(\Omega(\log^{\log_2 3} N))$ and $\exp(\Omega(\log^{2} N))$ respectively. It is easy to show that $N((\AND_2\circ \OR_2)^{\otimes k})$ and $N(\Maj_3^{\otimes k})$ is $\exp(O(n^{1/\log_2 3}))$ and $exp(O(n^{1/2}))$ respectively. So  showing optimality essentially boiled down to showing that the decision tree size of iterated $\AND_2\circ \OR_2$ and iterated $\Maj_3$ on $n$ variables is exponential $\exp(\Omega(n))$. This was established in \cite{jukna1999p} using spectral methods. We recover these size lower bounds using our rank lower bound for composed functions.
\begin{corollary}\label{corr:examples}
For $k\ge 1$ and $n=3^k$, 
$$\log \Size(\Maj_3^{\otimes k}) \ge \Rank(\Maj_3^{\otimes k}) \ge n/3 +1.$$
For $k\ge 1$ and $n=4^k$,
$$\log \Size((\AND_2\circ \OR_2)^{\otimes k}) \ge \Rank((\AND_2\circ \OR_2)^{\otimes k}) \ge n/4+1.$$
\end{corollary}
\begin{proof}
The $\Maj_3$ function has depth $3$ and rank $2$. Applying \cref{corr:iterated-rank}, we see that $\Rank(\Maj_3^{\otimes k})\ge 3^{k-1}+1=n/3 +1$. Since rank is a lower bound on $\log \Size$ (\cref{prop:rank_size}), we get the desired size lower bound for iterated $\Maj_3$. 
  
The $\AND_2\circ \OR_2$ function has depth $4$ and rank $2$. Again applying \cref{corr:iterated-rank}, we get $\log \Size((\AND_2\circ \OR_2)^{\otimes k})\ge \Rank((\AND_2\circ \OR_2)^{\otimes k})\ge 4^{k-1}+1=n/4+1$, giving the size lower bound for iterated $\AND_2\circ \OR_2$.
\end{proof}

The size lower bound for iterated $\AND_2\circ \OR_2$ on $n$ variables from \cite{jukna1999p}, in conjunction with the rank-size relation from \cref{prop:rank_size}, implies that the rank of the iterated function is $\Omega(n)$. \cref{corr:examples}  demonstrates that these tight rank and size lower bounds can be recovered simultaneously  \cref{corr:iterated-rank}. 

Recently (after the preliminary version of our work appeared), in \cite{cornelissen2022improved}, the rank of the iterated  $\AND_2\circ \OR_2$ function on $n$ variables was revisited, in the context of separating rank from randomised rank. Using the Prover-Delayer game-based characterisation of rank from \cref{thm:game-rank}, it was shown there that this function has rank exactly $(n+2)/3$. While an $\Omega(n)$ bound is now easy to obtain as in \cref{corr:examples}, getting the exact constants required significantly more work.

The arguments given in \cite{jukna1999p} and \cite{cornelissen2022improved} are tailored to the specific functions being considered, and do not work in general. On the other hand, our rank lower bound  from \cref{thm:compose-rank-bounds} implies  size lower bounds for composed functions in general. For completeness, we state our rank lower bound of \cref{thm:compose-rank-bounds} in terms of size.
\begin{corollary}
 For $k\ge 1$ and non-constant boolean functions $f$ and $g$,
\[  \log \Size(f\circ g) \ge \Depth(f) (\Rank(g)-1) +1 . \]
%And,
\[ \log \Size(f^{\otimes k}) \ge \Depth(f)^{k-1} (\Rank(f)-1) +1. \]
\end{corollary}

Using \cref{corr:examples}, we can now complete  the proof of \cref{lem:cert-not-ub}.
\begin{proof}(of \cref{lem:cert-not-ub})
  From \cref{{corr:examples}}, we know that $\Rank((\AND_2\circ \OR_2)^{\otimes k}) \ge n/4+1$.

  It is easy to see that $\CertZ((\AND_2\circ \OR_2))=\CertO((\AND_2\circ \OR_2))=2$, and that for $k>1$, $\CertZ((\AND_2\circ \OR_2)^{\otimes k}) = 2 \CertZ((\AND_2\circ \OR_2)^{\otimes k-1})$ and $\CertO((\AND_2\circ \OR_2)^{\otimes k}) = 2 \CertO((\AND_2\circ \OR_2)^{\otimes k-1})$. Thus $\CertZ((\AND_2\circ \OR_2)^{\otimes k})=\CertO((\AND_2\circ \OR_2)^{\otimes k})= 2^k=\sqrt{n}$.

  Hence, for $f=(\AND_2\circ \OR_2)^{\otimes k}$, $\Rank(f)\ge \frac{\CertZ}{2}\frac{\CertO}{2} + 1=\Cert(f)^2/4 +1$.  
\end{proof}

\subsection{$\CONJ$ decision trees}  
In this section, we consider a generalization of the ordinary decision tree model, namely $\CONJ$ decision trees. In the $\CONJ$ decision tree model, each query is a conjunction of literals, where a literal is a variable or its negation. (In \cite{KS-JCSS04}, such a tree where each conjunction involves at most $k$ literals is called a $k$-decision tree; thus in that notation these are $n$-decision trees. A 1-decision tree is a simple decision tree.) A model essentially equivalent to $\CONJ$ decision trees, $(\AND,\OR)$-decision trees, was investigated in \cite{benasher1995decision} for determining the complexity of $\Thr_{n}^{k}$ functions. $(\AND,\OR)$-decision trees query either an $\AND$ of a subset of variables or an $\OR$ of a subset of variables. It was noted in \cite{benasher1995decision} that the $(\AND,\OR)$ query model is related to the computation using Ethernet channels,  and  a tight query lower bound was shown in this model for $\Thr_{n}^{k}$ functions for all $k\ge 1$. It is easy to see that the $(\AND,\OR)$-query model is equivalent  to the $\CONJ$ query model upto a factor of 2. Let $\DepthGA$ and $\DepthAO(f)$ denote the query complexity of function $f$ in $\CONJ$ and $(\AND,\OR)$ query model respectively. Then
\begin{proposition}
For every boolean functions $f$,
\[ \DepthGA(f)\le \DepthAO(f) \le 2\DepthGA(f).\]
\end{proposition}
Such a connection is not obvious for the rank of $\CONJ$ and $(\AND,\OR)$ decision trees. 

Recently, in \cite{knop2021log}, a monotone version of $\CONJ$ decision trees called $\AND$ decision trees is studied, where queries are restricted to  $\AND$ of variables, not literals. To emphasize the difference, we refer to these trees as monotone $\AND$ trees.  Understanding monotone $\AND$ decision trees in \cite{knop2021log} led to the resolution of the log-rank conjecture for the class of $\AND$ functions (any function composed with the  2-bit $\AND$ function), up to a $\log n$ factor. As remarked in \cite{knop2021guest}, understanding these more general decision tree models has shed new light on central topics in communication complexity, including restricted cases of the log-rank conjecture and the query-to-communication lifting methodology.

In this section, we show that simple decision tree rank characterizes the query complexity in the $\CONJ$ decision tree model, up to a $\log n$ factor. Formally,
\begin{theorem}
\label{thm:simple-conj-relation}  
For every boolean functions $f$,
\[ \Rank(f) \le \DepthGA(f)\le \log \Size(f)\le \Rank(f)\log \left(\frac{e n}{\Rank(f)}\right).\]
Consequently, if $\Rank(f)=\Theta(n)$, then so is $\DepthGA(f)$.
\end{theorem}
\begin{proof}
First, we show that $\Rank(f) \le \DepthGA(f)$. This is the straightforward construction and its analysis. Let $T_f$ be a depth-optimal $\CONJ$ decision tree for $f$ of depth $d$. We give a recursive construction of an ordinary decision tree $T$ for $f$ of rank at most $d$. 
In the base case, when  $\Depth(T_f)=0$, set $T=T_f$.
In the recursion Step, $\Depth(T_f)\ge 1$. Let $Q$ be the literal-conjunction queried at the root node of $T_f$, and let $T_0$ and $T_1$ be the left and right subtree of $T_f$, computing $f_0$ and $f_1$ respectively. Recursively construct ordinary decision trees $T_0'$ and $T_1'$ for $f_0$ and $f_1$. Let $T_Q$ be the ordinary decision tree obtained by querying the variables in $Q$ one by one to evaluate the query $Q$.  Note that $T_Q$  evaluates the $\AND$ function on literals in $Q$; hence it has rank $1$, and has exactly one leaf labelled 1 . Attach $T_0'$  to each leaf labelled $0$ in $T_Q$, and and $T_1'$to the unique leaf labelled $1$ in $T_Q$, to obtain $T$. 

From the construction, it is clear that $T$ evaluates $f$. To analyse the rank of $T$, proceed by induction on $\Depth(T_f)$.
\begin{enumerate}
\item Base Case: $\Depth(T_f)=0$. Trivially true as $T=T_f$ with rank $0$.
\item Induction: $\Depth(T_f)\ge 1$. 
\begin{align*}
 \Rank(T) &\le \Rank(T_{Q}) + \max\{\Rank(T_0'),\Rank(T_1')\}
 \quad (\textrm{by \cref{prop:compose_rank_dt}})\\
 &= 1 + \max_{b\in\bool}\{\Rank(T_b')\} \\
 &\le 1 + \max_{b\in\bool}\{\DepthGA(f_b)\}   \quad (\textrm{by induction})\\
 &\le 1 + (\Depth(T_f)-1)= \Depth(T_f).
\end{align*}
\end{enumerate}

Next, we show that $\DepthGA(f)\le \log \Size(f)$. The main idea is that an ordinary decision tree of size $s$ can be balanced using $\CONJ$ queries into a $\CONJ$ decision tree of depth $O(\log s)$.

Let $T_f$ be a size-optimal simple tree for $f$ of size $s$. Associate with each node $v$ of $T_f$ a subcube $J_v$ containing all the inputs that reaches node $v$. The root node has the whole subcube $\boolset{n}$. For a node $v$, $J_v$ is defined by the variables queried on the path leading to the node $v$.
The recursive construction of a $\CONJ$ decision tree $T$ of depth at most $2 \log_{3/2} s$ proceeds as follows.
If $s=1$, set $T=T_f$.
Otherwise, in the
recursion step, $s> 1$. Obtain a node $v$ in $T_f$ such that number of leaves in the subtree rooted at $v$, denoted by $T_v$,  in the range $[s/3,2s/3)$. (Such a node necessarily exists, and can be found by starting at the root and traversing down to the child node with more leaves until the number of leaves in the subtree rooted at the current node satisfies the condition.) Let $J_v=(S,\rho)$ be the subcube associated with $v$, and let $Q$ be the $\CONJ$ query  testing membership in $J_v$;
$Q=(\bigwedge_{i\in S: \rho(i)=1} x_i)( \bigwedge_{i\in S: \rho(i)=0} \neg x_i)$.
%\[(\bigwedge_{\substack{i\in S \\ s.t. \rho(i)=1}} x_i )(\bigwedge_{\substack{i\in S \\ s.t. \rho(i)=0}} \neg x_i)\] 
%  $(\bigwedge_{\substack{i\in S \\ s.t. \rho(i)=1}} x_i)( \bigwedge_{\substack{i\in S \\ s.t. \rho(i)=0}} \neg x_i)$ which is $1$ if and only if the input lies in $J_v$.

  Note that $v$ cannot be the root node of $T$. 
  Let $u$ be the sibling of $v$ in $T_f$, and let $T_u$ be the subtree rooted at $u$. Let $w$ be the parent of $v$ and $w$ in $T_f$.
  Let $T'_f$ be the tree obtained from $T_f$ by removing the entire subtree $T_v$, removing the query at $w$, and attaching subtree $T_u$ at $w$. For all inputs not in $J_v$, $T'_f$ and $T_f$ compute the same value.

  $T$ starts by querying $Q$. If $Q$ evaluates to $1$, proceed by recursively constructing the $\CONJ$ decision tree for $T_v$. If $Q$ evaluates to $0$, proceed by recursively constructing the $\CONJ$ decision tree for $T'_f$.

The correctness of $T$ is obvious. It remains to estimate the depth of $T$. Let $D(s)$ be the number of queries made by the constructed $\CONJ$ decision tree. By construction, we have  $D(s)\le 1+ D(2s/3)$ giving us $D(s)=2\log_{3/2}s$, thereby proving  our claim.

The last inequality about size and rank, $\log \Size(f) \le \Rank(f)\log \left(\frac{e n}{\Rank(f)}\right)$, comes from \cref{prop:rank_size}. 
\end{proof}

\section{Tightness of Rank and Size relation for $\Tribes$}
\label{sec:rank-size}
In \cref{prop:rank_size}, we saw a relation between rank and size. The relationship is essentially tight. As remarked there, whenever $\Rank(f)=\Omega(n)$, the relation is tight. The function $f=\Parity_n$ is one such function that witnesses the tightness of both the inequalities. Since $\Rank(\Parity)=n$, \cref{prop:rank_size} tells us that $\log\Size(\Parity)$ lies in the range $[n, n \log e]$, and we know that $\log \Size(\Parity)=n$. 

For the $\Tribes_n$ function, which has $N=n^2$ variables, we know from \cref{thm:rank-tribes} that $\Rank(\Tribes_n)=n\in o(N)$. Thus \cref{prop:rank_size} tells us that $\log\Size(\Tribes_n)$ lies in the range $[n, n \log (en)]$.
(See also Exercise 14.9 \cite{Jukna-BFCbook} for a direct argument showing
$n \le \log\Size(\Tribes_n)$). 
But that still leaves a $(\log (en))$-factor gap between the two quantities.
We show that the true value is closer to the upper end. To do this,
we establish a stronger size lower bound for decision trees computing $\dTribes_n$.

\begin{lemma}\label{lem:size-lb-tribes}
  For every $n,m \ge 1$, every decision tree for $\dTribes_{n,m}$ has at least $m^n$ $1$-leaves and $n$ 0-leaves. 
\end{lemma}
\begin{proof}
  Recall that $\dTribes_{n,m} = \bigwedge_{i\in [n]} \bigvee_{j\in [m]} x_{i,j}$. We call $x_{i,1}, x_{i,2}, \ldots , x_{i,m}$ the $i$th block of variables. We consider  two special kinds of input assignments: 1-inputs of minimum Hamming weight, call this set $S_1$, and 0-inputs of maximum Hamming weight, call this set $S_0$.  Each $a\in S_1$ has exactly one 1 in each block; hence $|S_1|=m^n$. Each $b\in S_0$ has exactly $m$ zeroes, all in a single block; hence $|S_0|=n$. We show that in any decision tree $T$ for $\dTribes_{n,m}$, all the inputs in $S=S_1\cup S_0$ go to pairwise distinct leaves. Since all inputs in $S_1$ must go to 1-leaves of $T$, and all inputs of $S_0$ must go to 0-leaves, this will prove the claimed statement.

  Let $a,b$ be distinct inputs in $S_1$. Then there is some block $i\in[n]$, where they differ. In particular there is a unique $j\in[m]$ where $a_{i,j}=1$, and at this position, $b_{i,j}=0$.  The decision tree $T$  must query variable $x_{i,j}$ on the path followed by $a$, since otherwise it will reach the same 1-leaf on input $a'$ that differs from $a$ at only this position, contradicting the fact that $\dTribes_{n,m}(a')=0$. Since $b_{i,j}=0$, the path followed in $T$ along $b$ will diverge from $a$ at this query, if it has not already diverged before that. So $a,b$ reach different 1-leaves.
  
Let $a,b$ be distinct inputs in $S_0$. Let $i$ be the unique block where $a$ has all zeroes; $b$ has all 1s in this block. On the path followed by $a$, $T$ must query all variables from this block, since otherwise it will reach the same 0-leaf on input $a''$ that differs from $a$ only at an unqueried position in block $i$, contradicting $\dTribes_{n,m}(a'')=1$.  Since $a$ and $b$ differ everywhere on this block, $b$ does not follow the same path as $a$, so they go to different leaves of $T$.
\end{proof}

We thus conclude that the second inequality in \cref{prop:rank_size} is also asymptotically tight for the $\dTribes_n$ function.

The size lower bound from \cref{lem:size-lb-tribes} can also be obtained by specifying a good Delayer strategy in the asymmetric Prover-Delayer game and invoking  \cref{prop:game-size}.; see \cref{sec:game-proofs}.

\section{Proofs using Prover-Delayer Games}\label{sec:game-proofs}

In this section we give Prover-Delayer Game based proofs of our results. 

\paragraph*{Prover strategy for $\Tribes_{n,m}$, proving \cref{lem:rank-tribes-ub}}
We give a Prover strategy which restricts the Delayer to $n$ points, proving the upper bound on $\Rank(\Tribes_{n,m})$.

Whenever the Delayer defers a decision, the Prover chooses $1$ for the queried variable.

The Prover queries variables $x_{i,j}$ in row-major order.
In each row of variables, the Prover queries variables until some variable is set to 1 (either by the Delayer or by the Prover). Once a variable is set to 1,
the Prover moves to the next row of variables. 

This Prover strategy allows the Delayer to defer a decision for at most one variable per row; hence the Delayer's score  at the end is at most $n$.

\paragraph*{Delayer strategy for $\Tribes_{n,m}$, proving \cref{lem:rank-tribes-lb}}
We give a Delayer strategy which always score at least $n$ points, proving the lower bound.

On a query $x_{i,j}$, the Delayer defers the decision to the Prover
unless all other variables in row $i$ have already been queried.
In that case Delayer responds with a $1$. 

Note that with this strategy, the Delayer ensures that the game ends with function value $1$. (No row has all variables set to $0$.) 
Observe that to certify a $1$-input of the function, the Prover must query at least one variable in each row. Since $m\ge 2$, the Delayer gets to score at least one point per row, and thus has a score of at least $n$ at the end of the game.

\paragraph*{Prover strategy for $\AND_n \circ \Parity_m$, proving  \cref{lem:and-parity-ub}}
We give a Prover strategy which restricts Delayer to $n(m-1)+1$ points. The Prover queries variables in row-major order. If on query $x_{i,j}$ the Delayer defers a decision to the Prover, the Prover chooses arbitrarily unless $j=m$. If $j=m$, then the Prover chooses a value which makes the parity of the variables in row $i$ evaluate to $0$.

Let $j$ be the first row such that the Delayer defers the decision on $x_{j,m}$  to the Prover. (If there is no such row, set $j=n$.) With the strategy above, the Prover will set $x_{j,m}$ in such a way that the parity of the variables in $j$-th row evaluates to $0$, making $f$ evaluate to $0$ and ending the game. The Delayer scores at most $m-1$ points per row for rows before this row $j$, and at most $m$ points in row $j$.
Hence the Delayer's score is at most $(j-1)(m-1)+m$ points. Since $j\le n$, the Delayer is restricted to $n(m-1)+1$ points at the end of the game.

\paragraph*{Delayer strategy for $\AND_n \circ \Parity_m$, proving  \cref{lem:and-parity-lb}}
We give a Delayer strategy which always scores at least $n(m-1)+1$ points.

On query $x_{i,j}$, if this is the last un-queried variable, or if
there is some un-queried variable in the same $i$-th row, the Delayer defers
the decision to the Prover. 
Otherwise the Delayer responds with a value that makes the parity of the variables in row $i$ evaluate to $1$. 

This strategy forces the Prover to query all variables to decide the function.
The Delayer picks up $m-1$ points per row, and an additional point on the last query, giving a total score of $n(m-1)+1$ points.

\paragraph*{Prover strategy in terms of certificate complexity, proving  \cref{lem:rank-cert}}
We give a Prover strategy which restricts the Delayer to $ (\CertZ(f)-1)(\CertO(f)-1)+1$ points. Let $\tilde{f}$ be the function obtained by assigning values to the variables queried so far. As long as $\CertO(\tilde{f})>1$, Prover picks an $a\in \tilde{f}^{-1}(0)$ and its $0$-certificate $S$, and queries all the variables in $S$ one by one. If at any point the Delayer defers a decision to the Prover, the Prover chooses the value according to $a$. When $\CertO(\tilde{f})$ becomes $1$, the Prover picks an $a\in \tilde{f}^{-1}(1)$ and its $1$-certificate $\{i\}$ and queries the  variable $x_i$. If the Delayer defers the decision, the Prover chooses $a_i$.

The above strategy restricts the Delayer to $ (\CertZ(f)-1)(\CertO(f)-1) + 1$ points; the proof is essentially same as \cref{lem:rank-cert}.

\iffalse
\paragraph*{Delayer strategy for $f = \Maj_{2k+1}\circ \Maj_{2k+1}$, proving  \cref{lem:cert-not-ub}}
The following Delayer strategy always scores $(k+1)^2+k^2$ points, greater than $C(f)=(k+1)^2$.

At an intermediate stage of the game, say that a row is $b$-determined
if the variables that are already set in this row already fix the
value of $\Maj_{2k+1}$ on this row to be $b$, and is determined if it
is $b$-determined for some $b$.  Let $M_b$ be the number of
$b$-determined rows. If the game has not yet ended, then $M_0\le k$
and $M_1\le k$.

On query $x_{i,j}$, let $n_0,n_1$ be the number of
variables in row $i$ already set to $0$ and to $1$ respectively.  
The Delayer defers the decision if
\begin{itemize}
  \item row $i$ is already determined, or
  \item $n_0=n_1 < k$, or 
  \item  $n_0=n_1=k$ and $M_0=M_1$.
\end{itemize}
Otherwise, if $n_0\neq n_1$, then the Delayer chooses the value $b$
where $n_b < n_{1-b}$. If $n_0=n_1=k$, then the Delayer chooses the
value $b$ where $M_b < M_{1-b}$. 

This strategy ensures that at all stages until the game ends,
$|M_0-M_1|\le 1$, and furthermore, in all rows that are not yet
determined, $|n_0-n_1|\le 1$. Thus a row is determined only after
all variables in the row are queried, and the Delayer gets a point for
every other query, making a total of $k$ points per determined
row. Further, for $k+1$ rows, the Delayer also gets an additional
point on the last queried variable.  The game cannot conclude before
all $2k+1$ rows are determined, so the Delayer scores at least
$(k+1)^2+k^2$ points.
\fi
\paragraph*{Prover and Delayer strategies for composed functions, proving \cref{thm:compose-rank-bounds}}
For showing the upper bound, the Prover strategy is as follows: the  Prover  chooses a depth-optimal tree $T_f$ for $f$ and moves down this tree. Let $X^i$ denote the $i$th block of variables; i.e.\ the set of variables $x_{i,1}, x_{i,2}, \ldots , x_{i,m}$. The Prover queries variables blockwise, choosing to query variables from a particular block according to $T_f$.  If $x_i$ is the variable queried at the current node of $T_f$, the Prover  queries variables from $X^i$ following the optimal Prover strategy for the function $g$, until the value of $g(X^i)$  becomes known. At this point, the Prover moves to the corresponding subtree in $T_f$.  

For lower bound, the Delayer strategy is as follows: When variable $y_k$ is queried,
the Delayer responds with $b\in\bool$ if $\Rank(h_{k,b}) >
\Rank(h_{k,1-b})$, and otherwise defers. Here $h_{k,b}$ is the sub-function of $h$ when  $y_k$ is set to $b$.

The proof that above strategies give the claimed bounds is essentially what constitutes the proof of \cref{thm:compose-rank-bounds}.

\paragraph*{Delayer strategy in asymmetric game in $\dTribes_n$, proving \cref{lem:size-lb-tribes}}
We give a Delayer strategy in an asymmetric Prover-Delayer game which scores at least $n\log n$. On query $x_{ij}$,  Delayer responds with $(p_0,p_1)=(1-\frac{1}{k}, \frac{1}{k})$, where  $k$ is the number of free variables in row $i$ at the time of the query.

We show that the strategy above scores at least $n\log n$ points. The game can end in two possible ways:
\begin{enumerate}
\item Case 1: The Prover concludes with function value $0$. In this case, the Prover must have queried all variables in some row, say the $i$-th row, and chosen $0$ for all of them. For the last variable queried in  the $i$-th row, the Delayer would have responded with $(p_0,p_1)=(0,1)$, and hence scored $\infty$ points in the round and the game.
\item Case 2:  The Prover concludes with function value $1$. In this case, the Prover must have set a variable to $1$ in each row. We show that the Delayer scores at least $\log n$ points per row. Pick a row arbitrarily, and  let $k$ be the number of free variables in the row when the first variable in that row is set to $1$. The Prover sets $n-k$ variables in this row to $0$ before he sets the first variable to $1$. For $b\in\bool$, let $p_{b,j}$ represents the $p_b$ response of the Delayer when there are $j$ free variables in the row. That is, $p_{0,j}=1-\frac{1}{j}=\frac{j-1}{j}$ and $p_{1,j}=\frac{1}{j}$.  The contribution of this row  to the overall score is at least 
%% \begin{align*}
%% &=\sum_{j=n}^{k+1} \log \frac{1}{p_{0,j}} + \log \frac{1}{p_{1,k}} 
%% &=\log (\frac{1}{p_{0,n}}\frac{1}{p_{0,n-1}}...\frac{1}{p_{0,k+1}}\frac{1}{p_{1,k}})
%% &=\log (\frac{n}{n-1}\frac{n-1}{n-2}...\frac{k+1}{k}\frac{k}{1})\\
%% &=\log n
%% \end{align*}
%% \[\left(\sum_{j=n}^{k+1} \log \frac{1}{p_{0,j}}\right) + \log \frac{1}{p_{1,k}} 
%% = \log \left(\frac{1}{p_{0,n}}\frac{1}{p_{0,n-1}}...\frac{1}{p_{0,k+1}}\frac{1}{p_{1,k}}\right)
%% %= \log (\frac{n}{n-1}\frac{n-1}{n-2}...\frac{k+1}{k}\frac{k}{1})
%% =\log n .
%% \]
\[
\log \frac{1}{p_{0,n}} +
\log \frac{1}{p_{0,{n-1}}} + \ldots + 
\log \frac{1}{p_{0,k+1}} +
\log \frac{1}{p_{1,k}} 
= \log \left(\frac{1}{p_{0,n}}\frac{1}{p_{0,n-1}}...\frac{1}{p_{0,k+1}}\frac{1}{p_{1,k}}\right)
=\log n .
\]
Since each row contributes at least $\log n$ points, the Delayer  scores at least $n\log n$ points at the end of the game.
\end{enumerate}

\section{Conclusion}
\label{sec:concl}
The main thesis of this paper is that the minimal rank of a decision
tree computing a Boolean function is an interesting measure for the
complexity of the function, since it is not related to other
well-studied measures in a dimensionless way. Whether bounds on this
measure can be further exploited in algorithmic settings like learning or
sampling remains to be seen.  

\section{Acknowledgments}
The authors thank Anna G\'al and Srikanth Srinivasan for interesting
discussions about rank at the Dagstuhl seminar 22371.

\bibliographystyle{plain}
\bibliography{submission}

\appendix
\section{Proof of \cref{thm:rest-sparsity}}
\thmrestsparsity*
\begin{proof}
Let $J$ denote the subcube $(S,\rho)$, and let $\bar{S}$ denote the
  set $[n] \setminus S$.

  Partition the Fourier support of $f$ (sets $W$ for
  which$\hat{f}(W)\neq 0$) into buckets, one for each $T\subseteq
  \bar{S}$, as follows:
  \[ \textrm{Bucket for $T$,~} B_T = \{W \subseteq [n] : W\cap \bar{S}=T; \hat{f}(W) \neq 0\}.\]
  Note that $\sparn(f) = |B_\emptyset\setminus\{ \emptyset\}|+\sum_{T\subseteq \bar{S}; T \neq \emptyset} |B_T|$.

    In the Fourier expansion of $f$, group together terms with the same
  signature outside $S$. That is,
  \[f(x) =
  \sum_{T \subseteq \bar{S}} \left[\sum_{R \subseteq S} \hat{f}(T\cup R)\chi_{T\cup R}(x)\right]
    =
    \sum_{T \subseteq \bar{S}}
    \underbrace{\left[\sum_{R \subseteq S} \hat{f}(T\cup R)\chi_{R}(x) \right]}_{c_T(x)}\chi_T(x)
    =
    \sum_{T \subseteq \bar{S}} c_T(x) \chi_T(x)
  \]
  Inside  any subcube of the form $J'=(S,\rho')$, the functions $c_T(x)$ are independent of $x$, since $\chi_R|J' (x) = (-1)^{\sum_{i\in R} \rho'(i)}$.

  In particular, for the subcube $J$, $f|J$ is a constant, say $c$. Considering the Fourier expansion of $g=f|J$, we see that
  \[c = g = \sum_{T \subseteq \bar{S}} (c_T|J) (x) \chi_T(x).
  \]
  Since the Fourier representation of a function is unique, it follows that
  $c_\emptyset|J=c$, and $c_T|J = 0$ for all non-empty $T$. Fix any non-empty $T$. Then, by definition of $c_T$ and $B_T$, 
  \[(c_T|J)(x) = \sum_{R \subseteq S} \hat{f}(T\cup R)(\chi_{R}|J)
  = \sum_{W\in B_T} \hat{f}(W)(\chi_{R}|J) .\]
Since $c_T|J = 0$, it must have none or at least two non-zero terms
to achieve a cancellation. Thus for each non-empty $T$, if $B_T\neq
\emptyset$, then $|B_T|\ge 2$.

Now we can bound the sparsity of $f|J'$ for any subcube $(S,\rho')$.
Since $f|J'(x) %= \sum_{T \subseteq \bar{S}} (c_T|J')(x) \chi_T(x)
= \sum_{T \subseteq \bar{S}} (c_T|J') \chi_T(x)$, we
see that $\sparn(f|J')$ is at most the number of non-empty buckets
$B_T$ for non-empty $T$. Hence 
%Let $h=[B_\emptyset=\emptyset]$ be the indicator function, 1 if
%$B_\emptyset=\emptyset$ and 0 otherwise. Then
\begin{align*}
  \sparn(f|J') &= \textrm{number of non-empty $T$ with non-empty bucket $B_T$} \\
  &= \frac{1}{2} \sum_{T\subseteq \bar{S}, T\neq \emptyset, B_T \neq \emptyset} 2 \\
  &\le \frac{1}{2} \sum_{T\subseteq \bar{S}, T\neq \emptyset, B_T \neq \emptyset} |B_T| \\
  &\le \frac{1}{2}\sparn(f).
\end{align*}

\end{proof}

\end{document}